\documentclass[conference,compsoc]{IEEEtran}
\usepackage[T1]{fontenc}
\usepackage[nocompress]{cite}
\usepackage{amsmath,amssymb,amsthm,booktabs,array,graphicx,url}
\usepackage[hidelinks]{hyperref}
\hypersetup{pdftitle={Logit Leakage through Floating-Point Bias Addition},pdfauthor={Shaomin Chen, Xiaorui Gong}}
\newtheorem{theorem}{Theorem}
\newtheorem{proposition}{Proposition}

\newcounter{paperalgorithm}
\newcommand{\R}{\mathbb R}
\newcommand{\fl}{\operatorname{RN}_{32}}
\newcommand{\softmax}{\operatorname{softmax}}
\newcommand{\argmax}{\operatorname*{arg\,max}}
\graphicspath{{figures/}{./}}

\newif\ifreviewlinenumbers
\reviewlinenumbersfalse
\ifreviewlinenumbers
  \usepackage[switch]{lineno}
\fi
\begin{document}
\title{Logit Leakage through Floating-Point Bias Addition}
\author{\IEEEauthorblockN{Shaomin Chen\qquad Xiaorui Gong}}
\maketitle
\ifreviewlinenumbers\linenumbers\fi
\begin{abstract}
Language model APIs can expose token-specific biases while withholding logits and token probabilities. Existing bias-based extraction recovers relative logits, leaving a common additive component undetermined in ideal arithmetic. We identify a numerical side channel in floating-point bias addition that exposes information about this component through token identities alone. The channel distinguishes score vectors with identical complete unperturbed log probabilities. Our method exploits changes in representable spacing at binary exponent boundaries. Under the stated binary32 arithmetic assumptions and a known reachable score interval, it recovers a logit pair and its unknown tie-breaking priority exactly in at most 52 token queries. Given complete relative logits, one recovered anchor determines the original logit vector. We evaluate recovery across 11 public checkpoints from seven model families and extend the local method with adaptive interval search and anchor selection. In end-to-end experiments with vLLM, SGLang, and llama.cpp, adaptive recovery succeeds on all 128 Qwen3-0.6B inputs per system, including 64 held-out inputs. A separate token-only experiment extracts a pretrained small model's complete output-head representation up to an invertible hidden-basis transformation. We also establish a sufficient condition for suppressing common-shift disclosure by subtracting the maximum before bias addition. In controlled experiments, this mitigation eliminates all observed distinctions between common-shift variants and preserves token decisions under the tested random biases.
\end{abstract}

\section{Introduction}
Language-model services expose configurable decoding interfaces while withholding internal logits and model parameters. Logit bias lets users steer generated text by increasing or decreasing the likelihood of specified tokens, even when the interface returns only the selected token~\cite{openrouter2026parameters}. Understanding what numerical information can be recovered through these controls is part of assessing the information exposed by a prediction interface.

Prediction interfaces support model extraction~\cite{tramer2016stealing}, including recovery of language-model output projections and their dimensions~\cite{carlini2024stealing,finlayson2024logits}. Morris et al.\ and Carlini et al.\ recover relative logits through token-specific biases and observed token choices~\cite{morris2023inversion,carlini2024stealing}. In exact real arithmetic, both softmax probabilities and biased-argmax responses are invariant to a common additive shift of the logit vector. These observations determine logit differences while leaving one additive component unidentified for each input.

We show that floating-point bias addition can expose this component through token identities alone. Floating-point spacing varies with magnitude, so adding a chosen bias can produce different rounded comparisons for logit vectors related by a common shift. We construct queries that distinguish such vectors even when their complete unperturbed log probabilities are identical. Under explicit arithmetic and range conditions, we use these differences to recover original-logit anchors. Combined with complete relative logits, an anchor determines the original logit vector, including its common additive component.

Figure~\ref{fig:overview} summarizes the mechanism and recovery. We address three technical challenges:
\begin{itemize}
\item \textbf{Identifying leakage beyond relative logits.} Common-shift score vectors produce identical ideal query responses. Establishing additional disclosure requires floating-point queries that distinguish these vectors while their complete unperturbed probability outputs remain identical.
\item \textbf{Recovering original logits from token responses.} A token comparison combines score gaps, rounding, and an unknown tie-breaking order. Bounded biases restrict the accessible numerical regions, so recovery requires queries that separate these effects.
\item \textbf{Suppressing leakage while preserving bias control.} Removing dependence on the common additive component can change rounded comparisons and token choices. A mitigation must retain token-specific bias control and be evaluated for both channel suppression and preservation of decoding behavior.
\end{itemize}

\begin{figure*}[t]
\centering\includegraphics[width=\textwidth]{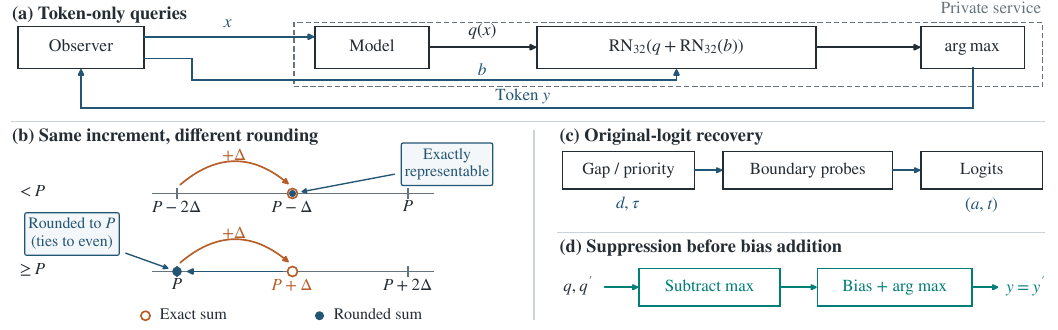}
\caption{Overview of the numerical channel and recovery. (a) The observer chooses token biases $b$ for fixed input $x$ and observes the selected token; $\mathrm{RN}_{32}$ denotes binary32 rounding. (b) At the positive boundary $P=2^e$, spacing increases from $\Delta$ to $2\Delta$. At the illustrated point below $P$, the increment is exact. At $P$, the exact sum $P+\Delta$ lies halfway between $P$ and $P+2\Delta$ and rounds to $P$ under ties-to-even. (c) Gap and priority recovery followed by boundary probes recovers the score pair under the stated arithmetic and range assumptions. (d) Subtracting the maximum before bias suppresses common-shift disclosure when stored variants yield identical centered vectors.}
\label{fig:overview}
\end{figure*}

For positive-score recovery, we determine the logit gap and the token tie-breaking priority, then probe the original anchor using the spacing change at an exponent boundary. The probe responses are periodic over part of the candidate range. We first use a small number of queries to determine the anchor's position within this repeating pattern, then restrict the candidate values accordingly to obtain a monotone comparison for binary search. Binary search then recovers the anchor exactly within the stated interval, and the recovered gap determines the second original logit.

Our central result assumes stable scores, independently converted biases, one correctly rounded addition per coordinate, greedy selection, and a fixed but unknown token tie-breaking order. Given a known reachable score interval, recovery takes at most 50 binary queries for a positive binary32 score pair and its priority, or 47 for negative scores. Selecting the two tokens takes two additional queries. We validate recovery across 11 public checkpoints from seven model families and extend the local method with interval search and anchor selection. End-to-end adaptive recovery through vLLM, SGLang, and llama.cpp succeeds on all 128 Qwen3-0.6B inputs per system, including 64 held-out inputs.

The recovered component has a direct role in output-head extraction. For a bias-free linear head under the rank conditions in Appendix~\ref{app:head}, centered logits determine $CWG$, where $C$ subtracts the vocabulary mean, $W$ is the output head, and $G$ is an invertible hidden-basis transformation. Original-logit anchors supply the missing shared row component and permit recovery of $WG$. We demonstrate this composition by extracting a pretrained model's complete $512\times64$ output-head representation from token responses across three independent input sets. On vocabulary rows excluded from basis alignment, weight RMSE ranges from $3.60\times10^{-7}$ to $1.20\times10^{-6}$ (Section~\ref{sec:head-evaluation}).

Finally, we analyze subtracting the maximum before bias addition as a mitigation. When common-shift variants yield identical preprocessed vectors, subsequent bias queries disclose no difference between them. Controlled experiments verify this condition and measure output preservation under independently generated sparse biases.

We make three contributions:
\begin{enumerate}
\item \textbf{Logit leakage through floating-point bias addition.} We characterize how rounding exposes a common additive component that remains unidentified by ideal probabilities and biased-argmax responses.
\item \textbf{Constructive recovery and output-head extraction.} We give local exact-recovery algorithms and query bounds, validate them through generation and serving interfaces, and combine recovered anchors with relative logits to extract an output-head representation.
\item \textbf{Processing-order analysis and mitigation.} We establish a sufficient condition for suppressing common-shift disclosure through pre-bias max subtraction and measure its effect on token decisions.
\end{enumerate}

\section{Background and Problem Formulation}\label{sec:setting}
\subsection{Logits, Probabilities, and Common Shifts}
At one generation position, input $x$ determines the vector $q(x)\in\R^V$ entering bias addition, where $V$ is the vocabulary size. We use \emph{scores} as a general term for numerical values used in token selection. A linear output head produces \emph{logits} $Wh(x)+\beta$ before softmax~\cite{vaswani2017attention,carlini2024stealing}. When bias addition acts directly on these outputs, we distinguish the \emph{original logits} from the resulting \emph{biased logits}. If temperature scaling or clipping precedes bias addition, our numerical analysis instead applies to the \emph{processed scores} entering that operation~\cite{huggingface2026generation}.

For a fixed input, write $q=q(x)$ and let $\mathbf1\in\R^V$ denote the all-ones vector. For any scalar $s\in\R$,
\begin{equation}
 \softmax(q+s\mathbf1)=\softmax(q).
 \label{eq:softmax-invariance}
\end{equation}
Writing the mean-centering projection as $C=I-\mathbf1\mathbf1^\top/V$ gives
\begin{equation}
 q=Cq+\mu\mathbf1,\qquad \mu=\mathbf1^\top q/V.
 \label{eq:decomposition}
\end{equation}
Here, $\mu$ is the mean score, and $Cq$ contains the deviations from that mean. Probabilities determine the centered component $Cq$ in exact arithmetic, but not $\mu$. An ideal biased-token interface
\begin{equation}
 O_q^{\mathrm{ideal}}(b)=\argmax_i(q_i+b_i)
 \label{eq:ideal}
\end{equation}
has the same ambiguity. This is the common freedom remaining in the logprob-free extraction setting of Carlini et al.~\cite{carlini2024stealing}.

Let $A$ be a selected token and $a=q_A$ its original score, which we call the anchor. Any recovered anchor fixes the common additive component: if an earlier extractor uses reference token $R$ and returns $d_i=q_i-q_R$, set $r_i=d_i-d_A$. Then $q_i=a+r_i$ and $\mu=a+\sum_i r_i/V$. The local protocol below recovers a score pair and its tie-breaking priority; Section~\ref{sec:anchor-selection} extends it to selecting the anchor.

The ambiguity also has an architectural interpretation. Replacing a linear output head $W$ by $W+\mathbf1v^\top$ changes the scores by $(v^\top h(x))\mathbf1$ without changing ideal probabilities. The lost scalar can therefore vary with the input and depend on a shared head component. Softmax-preserving output-head reparameterizations and their quantization effects have been studied for output-head optimization~\cite{kadav2026softmax}. We ask whether token responses to chosen biases can recover an original logit and thereby determine the missing mean $\mu$ when the complete relative logits are known.

\subsection{The Numerical Interface}\label{sec:contract}
We use the IEEE 754-2019 binary32 format (FP32), which has 24-bit significand precision~\cite{ieee2019floating}. For a fixed input, let $q$ be a vector of finite binary32 values. Let $\fl$ denote round-to-nearest, ties-to-even: halfway cases select the adjacent binary32 value with an even least significant significand bit~\cite{goldberg1991floating}. The observed interface is
\begin{equation}
 O_q(b)=\argmax_i\fl\!\left(q_i+\fl(b_i)\right).
 \label{eq:oracle}
\end{equation}
Each bias is converted independently before one binary32 addition. Scores remain fixed across queries for the same input. A fixed, initially unknown global tie-breaking order selects among tokens with equal maximum rounded scores. This selection is distinct from midpoint rounding; only the selected token identity is returned.

The recovery proofs use the separate bias conversion, rounded addition, and greedy selection in~\eqref{eq:oracle}. Conversion precision, inserted quantization, and fused operations can alter the comparisons~\cite{higham2002accuracy}. Section~\ref{sec:noise-design} extends the comparison protocol to specified errors in the returned tokens while keeping the underlying scores fixed.

A binary interval $[2^{e-1},2^e)$ is a \emph{binade}. Binary32 values in its normal range have spacing $\Delta=2^{e-24}$. At $P=2^e$, spacing doubles. We use this change to produce score-dependent rounding responses. The experiments use normal values far from overflow and underflow. The abstract extension to very small exponents additionally requires gradual underflow rather than flush-to-zero.

\subsection{Threat Model}\label{sec:threat}
\noindent\textbf{Protected information and adversary objective.}
The provider withholds the logits for each input. An external client seeks their original numerical values, including the common additive component left unidentified by ideal probabilities and biased argmax queries. Recovering one original anchor fixes this component when relative logits are known. Our local algorithm additionally recovers the target score $t$ and the pairwise tie-breaking priority; its exact-recovery criterion requires agreement with all three values at the specified decoder stage. Figure~\ref{fig:overview}(a) separates the client-visible interface from the private computation.

\noindent\textbf{Capabilities and observations.}
The client can repeat its own fixed input, request one greedy output token, and assign independent additive biases to selected token IDs within a public magnitude limit $H$. It adapts each bias to the preceding token responses. The selected token identity and tokenization are assumed observable. No returned probabilities, scores, score gaps, weights, hidden states, execution timings, co-location, model modification, or access to another user's request is required. The client does not choose the service's tie-breaking order, arithmetic implementation, or processing order.

\noindent\textbf{Numerical assumptions and prior knowledge.}
The service follows the stable-score operator in~\eqref{eq:oracle}, with an unknown fixed global tie-breaking order. Greedy selection alone does not establish score stability; Section~\ref{sec:noise} evaluates execution changes and extensions for noisy responses. The client knows this arithmetic model and a stated reachable score interval. The stated interval is required side information for the recovery theorem. In the single-window generation evaluation, separate calibration inputs select a public interval before testing; the test inputs do not supply scores to the observer. Precision, range, and stability are checked privately after predictions are frozen, and all fixed test inputs remain in the accounting.

The experiment evaluator can inspect private scores to measure success; the recovery controller cannot. In a search over multiple public intervals, uniqueness is defined relative to the prespecified search domain.

\section{Logit Leakage and Recovery}\label{sec:leakage}\label{sec:recovery}
\subsection{Leakage Mechanism}
Equation~\eqref{eq:ideal} makes every ideal query invariant under $q\mapsto q+s\mathbf1$. Common-shift variants therefore yield identical transcripts even for adaptive, randomized observers: conditioned on the observer's randomness, identical responses lead to identical subsequent queries. Floating-point addition need not preserve this equivalence.

We compare complete \emph{unperturbed} probability outputs with active numerical queries. An interface that returns probabilities after the same bias operations may also expose the additional information.

\noindent\textbf{A distinguishing query.}
Consider two binary32 vectors
\begin{equation}
 q=(30,28),\qquad q'=(29.875,27.875).
 \label{eq:witness-vectors}
\end{equation}
Their difference is a common shift. Their max-centered vectors are both $(0,-2)$, and their unperturbed probabilities are equal. Give the anchor bias $2$ and the target bias $4+3\cdot2^{-21}$. All supplied values are representable in binary32.

For $q$, the anchor becomes $32$. The target's exact sum is $32+3\cdot2^{-21}$, which rounds to $32$ because binary32 spacing above $32$ is $2^{-18}$. If the anchor has priority, it wins. For $q'$, the anchor becomes $31.875$. The target rounds to $31.875+2^{-19}$, because spacing below $32$ is $2^{-19}$, and wins strictly. Thus one numerical query distinguishes the two score vectors although their passive outputs are the same.

\begin{figure}[t]
\centering\includegraphics[width=\columnwidth]{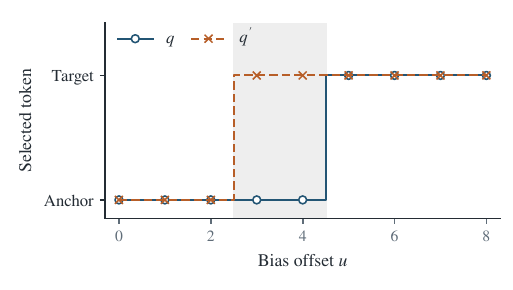}
\caption{Token responses for $q=(30,28)$ and $q'=q-\frac18\mathbf{1}$. The anchor bias is $2$; the target bias is $4+u2^{-21}$ for integer offsets $u$. The vectors have identical unperturbed probabilities. Shading marks offsets with different winners. The anchor wins ties between rounded scores; connecting lines guide the eye.}
\label{fig:mechanism}
\end{figure}

This example isolates the source of information. In exact arithmetic, the target would win for both vectors at the stated bias. Rounding distinguishes the two vectors, whose original relative scores are identical. Section~\ref{sec:defense-evaluation} evaluates how preprocessing affects these differences over complete vocabularies.

\subsection{Recovery Setting}
\noindent\textbf{Pair discovery and binary responses.}
An initial unbiased query selects the anchor, and a second query suppresses it by an allowed negative bias to select a target. The anchor is the original greedy winner, so $a\ge t$. Local recovery gives the anchor only nonnegative biases. No unmodified third token can then displace it under the fixed global tie-breaking order, making the subsequent response binary. Pair discovery is counted separately from the recovery bound.

Our recovery construction requires enough bias to move the two scores across a chosen binary exponent boundary $P=2^e$, where the spacing between adjacent binary32 values increases from $\Delta$ just below $P$ to $2\Delta$ just above it. Queries must also resolve the unknown tie-breaking priority: which token wins when the rounded scores are equal. We first specify a reachable score interval; Section~\ref{sec:consistency} addresses searches over multiple intervals and bounded-bias nonidentifiability.

Let $P=2^e$ with $-125\le e\le127$, assume gradual underflow and no overflow in probe arithmetic, and set $\Delta=P/2^{24}$. Define
\begin{equation}
 K=\left\lfloor\frac{\min(H,P/2)}{\Delta}\right\rfloor,\quad
 H_0=K\Delta,\quad F=P-H_0.
 \label{eq:range}
\end{equation}
For the anchor score $a$ and target score $t$, assume $K\ge4$ and
\begin{equation}
 a\ge t>F,\qquad
 a\in[\ell,u]\subset(F,P-2\Delta],
 \label{eq:local}
\end{equation}
where $\ell,u$ are $\Delta$ grid points. Both scores lie in the same binade, where binary32 values have spacing $\Delta$. Their gap $d=a-t$ is therefore a multiple of $\Delta$.

The interval in~\eqref{eq:local} provides enough bias to move the scores across $P$ while keeping every recovery bias between $0$ and $H_0$. The margin below $P$ also permits the adjacent probes used to resolve the priority and anchor residue. 

\subsection{Recovery Algorithm}
Algorithm~\ref{alg:recovery} implements the positive-score recovery sequence in Figure~\ref{fig:overview}(c). $\operatorname{FirstWinningIndex}$ finds the first target-winning bias index; $\operatorname{ResolveTie}$ identifies the gap and priority. $\operatorname{ResidueCandidates}$ selects anchors with the observed residue for $\operatorname{BinarySearch}$, using the monotone response $\operatorname{Probe}_{d,\tau}$. Write $O(b_0,b_1)$ for biases on tokens 0 and 1, with zero bias elsewhere. All subroutines use the public range parameters.

\par\smallskip
\noindent\begin{minipage}{\columnwidth}
\small
\refstepcounter{paperalgorithm}\label{alg:recovery}
\hrule\smallskip
\textbf{Algorithm~\thepaperalgorithm.} Exact recovery of a positive score pair.\par
\smallskip\hrule\smallskip
\textbf{Input:} Token oracle $O$, original winner 0, target 1, and public $P,H,\ell,u$ satisfying~\eqref{eq:local}; $\Delta,K,F$ follow~\eqref{eq:range}.\par
\textbf{Output:} Scores $(a,t)$ before bias addition and tie-breaking priority $\tau$.\par
\smallskip
\begin{tabular}{@{}r@{\quad}l@{}}
1 & $j\gets\operatorname{FirstWinningIndex}(O,\Delta,K)$\\
2 & $d_0\gets(j-1)\Delta$\\
3 & $\tau\gets\operatorname{ResolveTie}(O,d_0)$\\
4 & $d\gets d_0+\tau\Delta$\\
5 & $p_0\gets\max\{\ell,F+d+\Delta\}$\\
6 & $\mathcal S\gets\operatorname{ResidueCandidates}(O,d,\tau,p_0,u)$\\
7 & $a\gets\operatorname{BinarySearch}(\mathcal S,\operatorname{Probe}_{d,\tau})$\\
8 & $\textbf{return}\ (a,a-d,\tau)$
\end{tabular}
\smallskip\hrule
\end{minipage}
\par\smallskip

\noindent\textbf{Recovering the gap and tie-breaking priority.}
Label the original winner as token 0 and the target as token 1. Their pairwise tie-breaking priority is $\tau=0$ if token 0 wins when their rounded scores are equal, and $\tau=1$ otherwise. When only the target receives bias $k\Delta$, the first winning grid point $j$ determines exactly two possibilities:
\begin{equation}
 \begin{array}{ll}
 \tau=0:&d_0=(j-1)\Delta,\\
 \tau=1:&d_1=j\Delta=d_0+\Delta.
 \end{array}
 \label{eq:gap-hypotheses}
\end{equation}
The target loses at $k=0$ by the original-anchor definition and wins at $k=K$ because $t+H_0>P>a$. $\operatorname{FirstWinningIndex}$ finds $j$ by integer binary search in at most $\lceil\log_2K\rceil$ queries; the endpoint responses require no calls. Around the comparison threshold, the sums are representable grid values, so the search determines the two gap hypotheses exactly.

Three rounding probes identify the implementation's tie-breaking priority without assuming a token-ID order. Define
\begin{equation}
 p=\max\{\ell,F+d_0+\Delta\}.
\end{equation}
For $r=p,p+\Delta,p+2\Delta$, submit
\begin{equation}
 b_0=P-r,\qquad b_1=P-r+d_0+\Delta.
 \label{eq:tie-probes}
\end{equation}
Under the second hypothesis in~\eqref{eq:gap-hypotheses}, the two exact sums coincide and the target has priority, so all three replies are token 1. Under the first hypothesis, the responses follow the rounding pattern derived next and cannot all be token 1. $\operatorname{ResolveTie}$ returns $\tau=1$ exactly when all three replies are token 1, thereby resolving the gap ambiguity.

The equality case $a=t$ is included. It implies $\tau=0$ because token 0 was the original winner, and the first winning target bias is $\Delta$.

\noindent\textbf{Identifying the anchor residue.}
Suppose $d$ and $\tau$ are known. For a pivot $p\in\Delta\mathbb Z$ with $F+d+\Delta\le p\le P$, give the anchor bias $P-p$ and the target bias $P-p+d$. Their exact biased scores are equal:
\begin{equation}
 x=P+(a-p)=P+k\Delta,\qquad k=(a-p)/\Delta.
\end{equation}
Add one further $\Delta$ to whichever token has lower priority. After incorporating the increment, convert each bias independently to binary32 and add it to its score with one rounded addition.

Define $\operatorname{Probe}_{d,\tau}(p)$ as the indicator that token $1-\tau$ wins this query, and write $f(k)$ for its response at $k=(a-p)/\Delta$. Each evaluation makes one oracle call. Below $P$, both $x$ and $x+\Delta$ are representable and $f(k)=1$. At or above $P$, ties-to-even rounding operates on a grid of spacing $2\Delta$, yielding
\begin{equation}
 f(k)=
 \begin{cases}
  1,&k<0,\\
  0,&k\ge0,\ k\bmod4\in\{0,3\},\\
  1,&k\ge0,\ k\bmod4\in\{1,2\}.
 \end{cases}
 \label{eq:phase}
\end{equation}
The nonnegative pattern is $0,1,1,0$ (Figure~\ref{fig:residue-monotone}(a)). In particular, it contains no run of three ones. This establishes the priority test above because its initial pivot satisfies $p\le a$ under either gap hypothesis.

Set $p_0=\max\{\ell,F+d+\Delta\}$. Using at most three queries, $\operatorname{ResidueCandidates}$ either identifies $a=p_0$ or decodes the residue $\rho=((a-p_0)/\Delta)\bmod4$. It returns $\mathcal S=\{p_0\}$ in the former case and $\mathcal S=[p_0,u]\cap(p_0+\rho\Delta+4\Delta\mathbb Z)$ otherwise. A singleton interval needs no query. Appendix~\ref{app:boundary} gives the response table, the zero-offset exception, and the payload bounds.

\begin{figure}[t]
\centering
\includegraphics[width=\columnwidth]{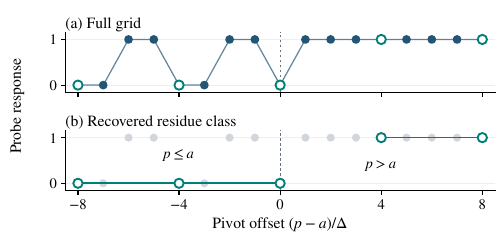}
\caption{Eliminating the periodic response. (a) Full-grid probes are periodic for $p\le a$ and constant for $p>a$. (b) The recovered residue class (open markers) gives a monotone comparison; gray points are excluded. The change at $p=a$ corresponds to floating-point spacing changing from $2\Delta$ to $\Delta$. Lines connect discrete queries. The horizontal origin uses the unknown $a$ for illustration.}
\label{fig:residue-monotone}
\end{figure}

Once the residue is known, restrict candidate anchors to that residue class, spaced $4\Delta$ apart. Both $a$ and a candidate pivot $p$ lie in this class, so $(a-p)/\Delta\in4\mathbb Z$. Equation~\eqref{eq:phase} therefore reduces to
\begin{equation}
 f((a-p)/\Delta)=1\quad\Longleftrightarrow\quad a<p.
 \label{eq:monotone}
\end{equation}
Figure~\ref{fig:residue-monotone} illustrates this restriction to a monotone comparison. $\operatorname{BinarySearch}$ queries $\operatorname{Probe}_{d,\tau}$ on $\mathcal S$, retaining candidates below the pivot after response 1 and those at or above it after response 0. It recovers $a$ in at most $\lceil\log_2|\mathcal S|\rceil$ queries, including zero for a singleton; then $t=a-d$.

\subsection{Correctness and Query Complexity}
\begin{theorem}[Local pair recovery]\label{thm:recovery}
Under~\eqref{eq:oracle} and~\eqref{eq:local}, with token 0 the original winner and a fixed global tie-breaking order, Algorithm~\ref{alg:recovery} recovers $(a,t,\tau)$ exactly. If the legal anchor interval after determining the gap contains $N=(u-p_0)/\Delta+1$ grid points, it uses at most
\begin{equation}
 \lceil\log_2K\rceil+6+
 \left\lceil\log_2\left\lceil N/4\right\rceil\right\rceil
 \label{eq:cost}
\end{equation}
binary recovery queries. For binary32 this is at most 50.
\end{theorem}
\begin{proof}
The target-only search gives the two gap hypotheses exactly. The three probes in~\eqref{eq:tie-probes} identify the correct hypothesis by~\eqref{eq:phase}. At most three more probes identify the anchor residue, including the $k=0$ exception. The remaining class contains at most $\lceil N/4\rceil$ points, and~\eqref{eq:monotone} permits ordinary binary search. The special case of a singleton interval terminates earlier. Every pivot keeps the anchor bias nonnegative and all payloads within $H_0$; the original-anchor property excludes third-token responses. Finally, $K\le2^{23}$ and $N\le K$, giving $23+3+3+21=50$.
\end{proof}

For binary precision $p\ge3$, the same full-binade construction gives a worst-case upper bound $2p+2$, under the corresponding representability assumptions.

The full-binade domain in~\eqref{eq:local} contains $(K-2)^2$ admissible pair-and-priority states. For $a=F+i\Delta$, with $1\le i\le K-2$, targets $t=F+j\Delta$, $1\le j\le i$, admit $2i-1$ states: two priorities for each $j<i$ and only the original-anchor priority for $j=i$. A binary decision tree of depth $T$ has at most $2^T$ leaves. Thus any deterministic algorithm that always recovers both scores and their priority has worst-case binary query cost at least
\begin{equation}
 \left\lceil\log_2(K-2)^2\right\rceil,
 \label{eq:lower}
\end{equation}
where $K=2^{p-1}$; for binary32, the bound is 46.

Existing multiclass extraction methods use more informative responses to recover many relative coordinates~\cite{carlini2024stealing}, with costs measured per coordinate.

\subsection{Negative Scores and Mixed Precision}\label{sec:extensions}
Logits can be negative. A negative exponent boundary reverses the direction of the spacing change and admits a separate construction without shifting the service's scores. Let $S=2^e$, $\delta=S/2^{23}$, $K=\lfloor\min(H,S)/\delta\rfloor$, $H_0=K\delta$, and $F=-S-H_0$. Assume $K\ge2$, representable $\delta/2$, no overflow or flushing of required subnormals, and
\begin{equation}
 a\ge t>F,\qquad a\in[\ell,u]\subset(F,-S-\delta].
 \label{eq:negative-local}
\end{equation}
The public endpoints lie on the $\delta$ grid.

\begin{theorem}[Negative-score recovery]\label{thm:negative}
Under~\eqref{eq:oracle} and~\eqref{eq:negative-local}, the score pair and unknown fixed tie-breaking priority can be recovered in at most
\begin{equation}
 \lceil\log_2K\rceil+2+
 \left\lceil\log_2\left\lceil N/2\right\rceil\right\rceil\le47
 \label{eq:negative-cost}
\end{equation}
binary recovery queries, where $N$ is the remaining number of legal anchor grid points after gap recovery. The worst-case binary-query lower bound for the full-binade pair-and-priority problem is 46.
\end{theorem}
The gap search is unchanged. One half-step query in the finer region resolves the tie-breaking priority. A further query at a pivot above every legal anchor identifies the anchor's parity. Restricting pivots to that parity turns the rounding response into a monotone comparison. Appendix~\ref{app:negative} proves the construction and its payload bounds. Changing to another known IEEE rounding direction permits modified probes with the same uniform upper bound; a mode change alone is therefore not a general suppression argument.

Known score precision restricts both searches to the storage grid. Storage-aligned probes also remove the fine-grid residue step, giving at most $2s-1$ recovery queries for $s$-bit stored scores when the bias and addition precision is at least $s+2$ and the local range and representability conditions hold: 15 for BF16 and 21 for binary16 with binary32 addition (Appendix~\ref{app:mixed}). The required grid is that of the scores entering bias addition; weight or forward-pass precision alone does not establish it. Lower-precision storage retains the channel when promoted scores accept finer biases.

\subsection{Consistency Checks and Abstention}\label{sec:consistency}
Outside a supplied local interval, we use local recovery to generate candidates over a fixed collection of windows. Each candidate must reproduce the combined transcript under the declared operator; zero or multiple surviving candidates cause abstention. A unique candidate is accepted relative to the specified search domain.

Bounded biases can leave distinct score vectors observationally equivalent (Appendix~\ref{app:deadband}). Thus, transcript consistency establishes uniqueness within the searched domain; exact recovery additionally requires the local arithmetic and range assumptions.

\subsection{Adaptive Anchor Selection}\label{sec:anchor-selection}
If no unique candidate survives the current pair's checks, discard its anchor and retry. Every query fixes bias $-H$ on previously discarded anchors $S$. A query with zero bias elsewhere selects anchor $A$; additionally suppressing $A$ selects target $T$. A response from $S$, or failure to select a distinct target, stops the search.

\noindent\textbf{Preservation of local recovery.}
The fixed mask defines a baseline vector with entries $\fl(q_i-H)$ on $S$ and $q_i$ elsewhere. Its greedy winner is $A$, so nonnegative anchor biases preserve binary responses under the fixed tie-breaking order. For $A,T\notin S$, the service computes $\fl(q_A+b_A)$ and $\fl(q_T+b_T)$ directly from the original logits. The local theorems therefore apply whenever the selected pair satisfies their range assumptions.

The search returns a unique candidate passing all checks or stops after $R$ pair attempts, biasing at most $R+1$ tokens per query. Window searches, checks, and retries contribute to the complete recovery cost.

\subsection{Output-Head Extraction}\label{sec:head-method}
Building on the output-head factorization of Carlini et al.~\cite{carlini2024stealing}, we combine complete relative logits with recovered anchors to determine a bias-free linear head up to an invertible hidden-basis transformation. Bias bisection against the unbiased winner estimates each relative coordinate (Appendix~\ref{app:head}). Centering these estimates across the vocabulary gives a matrix that we factor as $\widehat B\widehat X$. The anchors determine the vector $\boldsymbol\mu$ of mean logits over extraction inputs. Fitting $\boldsymbol\mu=u\widehat X$ gives the shared row $u$, and adding it to every vocabulary row completes the head as $\widehat W=\widehat B+\mathbf1u$.

If the centered head and hidden-state matrix have full hidden rank, these factors represent $CWG$ and $WG$ for an invertible $G$; Appendix~\ref{app:head} derives the identity and handles rank loss under centering. Relative-coordinate and floating-point errors make the numerical head estimates approximate even when the anchors are exact. For a held-out input, queried relative coordinates determine its coordinates in the learned basis; $u$ predicts the mean without using that input's anchor.

\section{Experimental Setup}\label{sec:methodology}
\subsection{Models and Interfaces}
We evaluate 11 checkpoints from seven families, ranging from 14M to 4B parameters: Pythia~\cite{biderman2023pythia}, GPT-2~\cite{radford2019gpt2}, GPT-Neo~\cite{black2021gptneo}, OPT~\cite{zhang2022opt}, BLOOM~\cite{bigscience2022bloom}, SmolLM2~\cite{benallal2025smollm2}, and Qwen3~\cite{qwen2025report}. Generation experiments invoke unmodified \texttt{model.generate} for one greedy token with the \texttt{SequenceBiasLogitsProcessor} in Transformers 4.57.1~\cite{wolf2020transformers,huggingface2026generation}. Forward computation uses FP32 or BF16 as specified below; bias addition uses FP32.

Serving experiments use Qwen3-0.6B and SmolLM2-360M on research-controlled vLLM 0.10.2~\cite{kwon2023pagedattention}, SGLang 0.5.3.post2~\cite{zheng2024sglang}, and \texttt{llama-server} at commit \texttt{7f2dd88b0ac3}~\cite{llamacpp2026server}. GPU servers use BF16 weights; llama.cpp uses CPU FP32 GGUF. All add binary32 biases. Output-head extraction uses pretrained TinyStories-260K~\cite{karpathytinyllamas}, whose bias-free head has a complete 512-token vocabulary, through a research-controlled CPU FP32 HTTP service.

\subsection{Observation and Validation Protocol}
The observer receives token IDs; evaluator-only hooks record original logits as ground truth. End-to-end runs obtain every response through generation calls, while cached evaluations apply the decoder to saved full-vocabulary vectors behind the same token-only interface. Predictions are frozen before private scoring. Independent replay checks recorded responses, and selected serving traces are repeated without instrumentation. Calibration uses separate inputs; precision, budget, and framework comparisons reuse the specified input sets and count their model--input conditions separately. Each study below gives its input split and comparison protocol.

\subsection{Evaluation Metrics}\label{sec:metrics}
For true scores and tie-breaking priority $(a,t,\tau)$ and a returned estimate $(\hat a,\hat t,\hat\tau)$, we define the maximum absolute score error and exact-recovery criterion as
\begin{equation}
\begin{aligned}
 E_{\max}&=\max\{|\hat a-a|,\ |\hat t-t|\},\\
 \text{Exact recovery}&\iff E_{\max}=0\ \land\ \hat\tau=\tau.
\end{aligned}
\label{eq:success}
\end{equation}
For the finite, nonzero scores in our recovery domains, score equality is bitwise equality of the binary32 representations. Even a one-ULP discrepancy is counted as incorrect. We report recovered, incorrect, and unresolved cases over all attempted inputs, as well as exact recovery conditional on the stated range assumptions.

Adaptive anchor recovery requires $\hat a=q_A$ for the returned token $A$; we also evaluate the selected pair under~\eqref{eq:success}, retaining all attempted inputs.

A \emph{recovery query} is one call used after token-pair selection; the two selection calls are counted separately. A \emph{generation call} recomputes the model and applies its decoding operations. A \emph{cached decoder evaluation} applies those operations to a saved score vector. We report these computational costs separately.

\section{Evaluation}\label{sec:evaluation}
We address four research questions:
\begin{itemize}
\item \textbf{RQ1:} Can token responses recover original logits across models and numerical precisions, and how does range coverage affect recovery? (Section~\ref{sec:recovery-evaluation})
\item \textbf{RQ2:} Can the adaptive procedure recover logits through serving APIs, and what is its complete query cost? (Section~\ref{sec:native-evaluation})
\item \textbf{RQ3:} Can recovered anchors extend token-only output-head extraction from $CWG$ to $WG$? (Section~\ref{sec:head-evaluation})
\item \textbf{RQ4:} How do execution changes and token-response errors affect recovery? (Section~\ref{sec:noise})
\end{itemize}

\subsection{RQ1: Recovery across Models and Precisions}\label{sec:recovery-evaluation}\label{sec:generation-design}
The FP32 evaluation uses the same 64 fixed test inputs per checkpoint. Eight prefixes, including Chinese, JSON, and always-A instructions, are balanced across arithmetic and word-continuation text to exercise the numerical interface. The single-window baseline decodes under one evaluator-calibrated interval; the adaptive controller searches windows, checks the combined transcript, and reselects anchors as defined in Sections~\ref{sec:consistency}--\ref{sec:anchor-selection}.

Before testing, eight separate calibration inputs select one public positive or negative interval by maximum calibration coverage, breaking coverage ties by ascending exponent and then positive before negative. The candidate boundaries are $\pm2^e$ for $e=-6,\ldots,10$, with the theorem's endpoint margins and primary bias limit $H=100$. Calibration uses evaluator score access to supply the assumed interval; it is not a token-only range-estimation procedure. No test score reaches the observer, and every fixed input is retained.

A post-hoc cached evaluation compares all 34 signed windows, first with the original pair and then with selection capped at eight pairs. It also evaluates larger bias budgets and generic versus known-precision recovery on the same vectors (Appendix~\ref{app:anchor-selection}). End-to-end adaptive runs use 64 FP32 GPT-Neo inputs and 32 BF16 inputs each for OPT and Qwen3 at $H=100$. Each query recomputes the model and returns only the selected token.

\begin{table*}[t]
\centering\small
\caption{FP32 recovery with a calibrated interval and adaptive anchor selection.}
\label{tab:evidence}
\begin{tabular*}{\linewidth}{@{\extracolsep{\fill}}lrrrrrrr@{}}\toprule
 & & \multicolumn{3}{c}{Generation} & \multicolumn{3}{c}{Cached multi-window search}\\
\cmidrule(lr){3-5}\cmidrule(l){6-8}
Model & $H$ & Boundary & Exact & Calls & Boundaries & Fixed & Selected\\\midrule
Pythia-14M & 100 & $16$ & 64 & 3,284 & $16$ & 64 & 64\\
Pythia-160M & 100 & -- & 0 & 1,600 & -- & 0 & 0\\
Pythia-1.4B & 100 & $16$ & 64 & 3,269 & $16$ & 64 & 64\\
GPT-2 Medium & 100 & $-64$ & 54 & 3,136 & $-64, -32$ & 62 & 64\\
GPT-Neo-125M & 100 & $-1$ & 13 & 2,363 & $-2, -1, -0.5$ & 19 & 64\\
OPT-1.3B & 100 & $16$ & 58 & 3,273 & $8, 16$ & 59 & 64\\
BLOOM-560M & 100 & $512$ & 56 & 3,095 & $512$ & 56 & 56\\
SmolLM2-360M & 100 & $32$ & 64 & 3,288 & $32$ & 64 & 64\\
SmolLM2-1.7B & 100 & $16$ & 32 & 3,279 & $16, 32$ & 64 & 64\\
Qwen3-0.6B & 100 & $32$ & 63 & 3,254 & $32$ & 63 & 64\\
Qwen3-4B-Instruct & 100 & $64$ & 36 & 3,143 & $32, 64$ & 50 & 64\\
\midrule\multicolumn{8}{l}{Larger-budget ablation on the same inputs}\\
Pythia-160M & 256 & $1024$ & 64 & 3,156 & $1024$ & 64 & 64\\
BLOOM-560M & 160 & $512$ & 64 & 3,198 & $512$ & 64 & 64\\
\bottomrule\end{tabular*}
\par\smallskip
\begin{minipage}{\linewidth}\footnotesize
\textit{Notes.} Counts are out of 64 inputs. Generation uses one evaluator-calibrated interval; Calls exclude calibration. Its 704 primary cases include 103 incorrect and 97 unresolved outcomes. Cached Boundaries list accepted windows (--: none), from $\pm2^e$, $e=-6,\ldots,10$. Fixed uses one pair; Selected tries up to eight. Accepted anchors meet exact pair-and-priority recovery; other cached outcomes are unresolved. Pythia-160M at $H=100$ has zero calibration coverage (--); its 1,600 calls use the fallback $P=2^{-6}$.
\end{minipage}
\end{table*}

Across 32,984 generation calls on 704 FP32 cases, all 504 covered pairs and priorities recover exactly: 504/704 (71.6\%) overall and 504/504 within the calibrated intervals (Table~\ref{tab:evidence}). GPT-2 and GPT-Neo exercise negative recovery on 54 and 13 covered inputs. Every response matches full-vocabulary binary32 replay, and original-logit hashes remain stable within each input, validating the tested generation path in~\eqref{eq:oracle}.

Outside the calibrated interval, the single-window baseline returns 103 incorrect candidates and leaves 97 unresolved. Cached window search recovers 565/704 inputs with the original pair; adaptive selection increases this to 632/704 (89.8\%) with zero incorrect accepted candidates. Nine checkpoints reach 64/64, using at most three pairs. All 72 remaining inputs exhaust eight attempts outside the tested local domains.

The larger-budget ablation reaches 64/64 for Pythia-160M at $H=256$ and BLOOM at $H=160$, bringing cached adaptive coverage to 704/704. BLOOM's eight unresolved $H=100$ targets require $512-t\approx106.77$--$143.17$ bias, explaining their exclusion by that budget. Table~\ref{tab:evidence} retains both budgets; the artifact contains the full budget sweep.

With all 34 windows and anchor selection, end-to-end GPT-Neo recovery reaches 64/64 at $H=100$; 45 inputs require another anchor. The complete procedure averages 2,428 generation calls per input, including selection, window searches, checks, and retries. This cost measures the broad search beyond the local recovery bound.

\noindent\textbf{BF16 logits.}
Four checkpoints additionally use BF16 forward computation: Pythia-160M, OPT-1.3B, BLOOM-560M, and Qwen3-0.6B. For each, the same 32 inputs are evaluated with the binary32 recovery algorithm and the algorithm using the known logit precision; the generation processor still adds FP32 biases. We verify the logit lattice privately. A separately specified budget ablation repeats configurations with zero calibration coverage at $H=100$ using $H=256$. Selection depends on calibration, not test success. It uses the same test inputs and keeps the original-budget outcomes.

The channel persists with BF16 logits and FP32 bias processing in the tested configurations. At $H=100$, both controllers recover all 88 covered inputs out of 128; the larger-budget Pythia-160M intervention adds 32 covered paired cases. Complete original logit vectors, token pairs, and rendered inputs match bit for bit between controllers. The original logits lie on a BF16 lattice, as required by Proposition~\ref{prop:mixed}.

\begin{table}[t]
\centering\small
\caption{BF16 recovery with FP32 bias processing.}
\label{tab:mixed}
\begin{tabular*}{\linewidth}{@{\extracolsep{\fill}}lrrrrr@{}}\toprule
 & & \multicolumn{3}{c}{Generation} & Cached\\
\cmidrule(lr){3-5}\cmidrule(l){6-6}
Model & $H$ & Exact & Generic & Known & Selected\\\midrule
Pythia-160M & 100 & 0 & -- & -- & 0\\
Pythia-160M & 256 & 32 & 48 & 13 & 32\\
OPT-1.3B & 100 & 30 & 50 & 14--15 & 32\\
BLOOM-560M & 100 & 28 & 46--48 & 11--13 & 28\\
BLOOM-560M & 160 & NR & NR & NR & 32\\
Qwen3-0.6B & 100 & 30 & 49--50 & 14--15 & 32\\
\bottomrule\end{tabular*}
\par\smallskip
\begin{minipage}{\linewidth}\footnotesize
\textit{Notes.} Counts are out of 32; both algorithms agree. Generic and Known columns give generic and BF16-aware recovery queries, excluding pair selection. Both use FP32 bias addition. Cached selection uses multiple windows and up to eight pairs, with zero incorrect returns. --: no covered case; NR: generation not run.
\end{minipage}
\end{table}

Table~\ref{tab:mixed} shows 11--15 local queries with known BF16 precision versus 46--50 with generic recovery. Cached anchor selection covers 128/128 inputs at the displayed larger budgets. In end-to-end adaptive runs at $H=100$, OPT reaches 32/32 with both algorithms on bitwise-identical original logits. Mean complete costs are 439 calls with known precision and 1,265 with generic recovery. Qwen3 reaches 32/32 with generic recovery at 1,060 calls per input. Each model needs a second pair on two inputs.

\subsection{RQ2: Recovery in Serving Systems}\label{sec:native-evaluation}\label{sec:native-design}
Each model--framework configuration uses 64 fixed inputs, $H=100$, and one interval calibrated on eight separate inputs with evaluator-only scores. An independent client submits biases through the generation API under fixed execution settings. Read-only probes record evaluation scores; every eighth input's complete trace is rerun on a fresh uninstrumented server. Appendix~\ref{app:repro} specifies the interface fields.

Cached ablations on these 384 vectors compare $P=16$, $P=32$, and both, with up to eight pairs and six candidate checks. We freeze the two-window controller before evaluating it end to end on Qwen3 through each API, using the original 64 inputs and 64 disjoint held-out inputs per framework. All payloads, including repeats, are submitted to the server.

The single-window baseline at $P=32$ recovers exactly 301/384 inputs: 37/36/36 Qwen3 cases across vLLM/SGLang/llama.cpp and all 192 SmolLM2 cases (Table~\ref{tab:native}). The other 83 return incorrect candidates outside the calibrated interval; none is unresolved (Section~\ref{sec:range-discussion}). Covered cases use 49--50 recovery queries plus two selection calls. All 19,686 responses match binary32 replay, with stable original logits; 2,447 uninstrumented-server queries reproduce their tokens.

\begin{table}[t]
\centering\small
\caption{End-to-end recovery in LLM serving systems.}
\label{tab:native}
\begin{tabular*}{\linewidth}{@{\extracolsep{\fill}}lrrrrr@{}}\toprule
 & \multicolumn{2}{c}{Single window} & \multicolumn{3}{c}{Adaptive Qwen3}\\
\cmidrule(lr){2-3}\cmidrule(l){4-6}
Framework & Qwen3 & SmolLM2 & Dev. & Test & Calls\\\midrule
vLLM & 37 & 64 & 64 & 64 & 119.6\\
SGLang & 36 & 64 & 64 & 64 & 120.4\\
llama.cpp & 36 & 64 & 64 & 64 & 115.4\\
\bottomrule\end{tabular*}
\par\smallskip
\begin{minipage}{\linewidth}\footnotesize
\textit{Notes.} Recovery counts are out of 64; $H=100$. Models: Qwen3-0.6B and SmolLM2-360M. Single window uses $P=32$. Adaptive uses $P=16,32$, checks, and up to eight pairs. Dev./Test denote the development and disjoint held-out inputs. Calls are mean complete API calls across both adaptive sets, including selection, searches, checks, and retries; calibration and independent replay are excluded.
\end{minipage}
\end{table}

Single-window checks still accept seven incorrect Qwen3 candidates on llama.cpp: four at $P=16$ and three at $P=32$. All lie outside the assumed local domain and contradict the other window's transcript. Combining both windows recovers 58/57/56 of the 64 development inputs across the three frameworks with the original pair; anchor selection raises each to 64/64.

The frozen two-window controller recovers 64/64 development and 64/64 held-out inputs per framework: 384/384 exact pairs and priorities, with no incorrect or unresolved returns. Another anchor is needed in 39 cases. Complete cost averages 118.5 API calls per input condition; per-framework medians are 109/109/108 and maxima are 310/312/215. All 45,497 responses match binary32 replay, and 6,271 uninstrumented-server queries reproduce their tokens. Appendix~\ref{app:anchor-selection} gives the search-cost breakdown.

\subsection{RQ3: Output-Head Extraction}\label{sec:head-evaluation}\label{sec:head-design}
We evaluate the completion in Section~\ref{sec:head-method} on three independently generated input sets for TinyStories-260K. Sets 1 and 3 contain 96 extraction inputs and 16 held-out inputs; set 2 contains 192 and 32. Random token prefixes have lengths 4, 8, and 16. Each HTTP request recomputes 64 copies of the input, applies independent biases, and returns token IDs; padding fixes the batch shape. The observer uses $H=100$, boundaries $P\in\{1,2,4,8,16,32\}$, and at most four pair attempts.

The observer receives no hidden width. Rank selection uses the largest adjacent singular-value ratio, requiring a ratio of at least 10, a preceding singular value above $10^{-7}$ of the largest, and two input columns beyond the selected rank.

Separate programs collect observations, extract factors, and validate frozen results. The centered and completed factors share one basis alignment fitted on 256 vocabulary rows, centered using those rows; head errors are evaluated on the remaining 256. Held-out predictions use queried relative coordinates and $u$; separately recovered test anchors are used only for evaluation. A post-hoc constant baseline predicts the average of the extraction inputs' recovered means. An input-budget analysis refactors nested extraction prefixes with the held-out inputs fixed.

On TinyStories-260K, all three random-prefix input sets yield complete relative-coordinate estimates and exact anchors, pairs, and priorities for every extraction and held-out input. Each full extraction identifies rank 64, matching the hidden width. Table~\ref{tab:head-extraction} reports head reconstruction and prediction of held-out mean logits from their queried relative coordinates.

\begin{table*}[t]
\centering\small
\caption{Token-only extraction of TinyStories-260K's complete $512\times64$ output head and held-out mean prediction.}
\label{tab:head-extraction}
\begin{tabular*}{\textwidth}{@{\extracolsep{\fill}}crrrrrr@{}}\toprule
Input set & Extraction/test & Exact anchors & Relative head error & Mean RMSE & Token decisions & HTTP batches\\\midrule
1 & 96/16 & 96/96 & $3.96\times10^{-6}$ & $3.11\times10^{-7}$ & 1,370,786 & 45,873\\
2 & 192/32 & 192/192 & $1.19\times10^{-6}$ & $4.59\times10^{-7}$ & 2,740,717 & 91,584\\
3 & 96/16 & 96/96 & $3.35\times10^{-6}$ & $7.21\times10^{-7}$ & 1,372,446 & 46,928\\
\bottomrule\end{tabular*}
\par\smallskip
\begin{minipage}{\textwidth}\footnotesize
\textit{Notes.} Independent input sets use the same checkpoint. Head error is relative Frobenius error on 256 vocabulary rows excluded from basis alignment. Mean-logit RMSE evaluates prediction from held-out relative coordinates without their anchors. Anchor counts and costs cover extraction inputs only; all 64 separately evaluated test anchors are also exact. HTTP requests batch independently biased decisions. Probability queries: zero.
\end{minipage}
\end{table*}

\noindent\textbf{Recovery of the shared component.}
Under the shared basis alignment, the completed head has relative Frobenius error $1.19\times10^{-6}$--$3.96\times10^{-6}$. The relative branch recovers its own target $CW$ with RMSE $3.50\times10^{-7}$--$1.10\times10^{-6}$, but differs from the original $W$ by RMSE $0.189$ because it omits the shared row. The recovered shared row has RMSE $8.25\times10^{-8}$--$4.32\times10^{-7}$. Figure~\ref{fig:head-extraction}(a) shows the entrywise absolute errors before and after this completion. The extracted representation is $WG$, retaining the invertible hidden-basis ambiguity.

For the 64 held-out inputs, the learned map predicts the mean from relative coordinates with RMSE $3.11\times10^{-7}$--$7.21\times10^{-7}$. Predicting the average extraction-input mean instead gives RMSE $0.895$--$1.162$, showing the variation captured by the map. Combining the predicted means with relative coordinates gives complete original-logit RMSE $2.15\times10^{-6}$--$3.60\times10^{-6}$. Appendix~\ref{app:head} reports numerical diagnostics.

\begin{figure*}[t]
\centering\includegraphics[width=\textwidth]{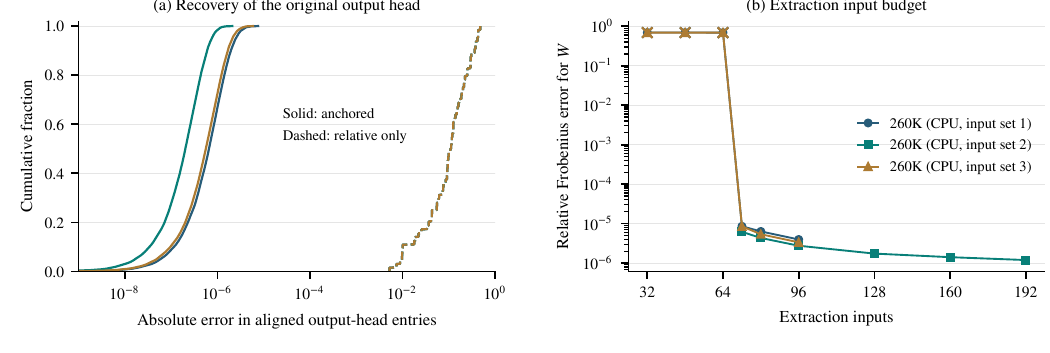}
\caption{Token-only output-head extraction from TinyStories-260K across three independently generated input sets. (a) Empirical distributions of absolute entry errors against the original $W$, after the same basis alignment. Solid curves include the recovered shared component; dashed curves omit it. The relative branch's own target is $CW$. (b) Output-head relative error as the extraction input count increases, with each set's test inputs fixed. Crosses mark inferred ranks below 64; set 2 is unresolved at 32--64 inputs and has no corresponding error points. Lines connect measured counts.}
\label{fig:head-extraction}
\end{figure*}

\noindent\textbf{Input budget and query cost.}
At 32--64 inputs, the rank rule selects rank 1 or remains unresolved: its full-rank candidate needs more than 64 input columns. All three sets select rank 64 at 72 inputs. For set 2, increasing the budget from 96 to 192 reduces relative head error from $2.75\times10^{-6}$ to $1.19\times10^{-6}$ (Figure~\ref{fig:head-extraction}(b)). Across the three sets, mean cost per extraction input is about 14,300 token decisions: 255--267 for the anchor and 14,019--14,029 for relative coordinates. Anchors account for about 1.8\% of decisions; their sequential queries and the batched relative-coordinate probes have different latency costs. Table~\ref{tab:head-extraction} separates extraction cost from held-out evaluation.

\subsection{RQ4: Execution Stability and Response Noise}\label{sec:noise}\label{sec:noise-design}
\noindent\textbf{Stability of the numerical target.}
Greedy decoding can produce execution-dependent floating-point scores~\cite{he2025nondeterminism}. We test Pythia-14M and Pythia-160M with eight new inputs each, CPU FP32 inference, and batch sizes $1,2,4,8$. Each input--batch condition repeats six probes ten times: no bias, early and final gap search, priority identification, residue identification, and final anchor search. We also repeat complete recovery at fixed batch sizes 1 and 4.

Across 3,840 generation calls, unmodified greedy tokens remain unchanged, but the selected Pythia-14M and Pythia-160M scores move by up to 5 and 2 ULPs across batch conditions. Scores remain bitwise stable within each fixed input--batch condition. Complete recovery succeeds on 16/16 inputs at each of batch sizes 1 and 4, although 15/16 recovered pairs differ between conditions. Fixing execution conditions defines a stable recovery target in this study.

\noindent\textbf{Errors in responses about a fixed target.}
To isolate response noise from score changes, we inject independent token-response flips with probability $p_{\mathrm{flip}}\in\{0,0.01,0.05,0.1,0.2,0.3,0.4\}$ into the 504 previously covered FP32 cases. Five noise seeds per case give 2,520 simulated recoveries per method and noise level. Token pairs and public intervals are fixed before injection; the 200 out-of-range cases and pair discovery are outside this conditional study.

We compare four binary-response rules. \emph{Single} uses one response; \emph{Majority-9} votes over nine. The sequential likelihood-ratio rule (\emph{LR}) uses error bound $\bar p_{\mathrm{flip}}=\max(0.1,p_{\mathrm{flip}})$~\cite{wald1945sequential}. The confidence-interval rule (\emph{CI}) samples until an interval excludes $1/2$, without a noise-rate bound, allocating error across predetermined sample sizes~\cite{howard2021timeuniform}. Both use recovery-wide error budget $\delta=0.01$ for at most $M=50$ logical comparisons. Each comparison is capped at 4,096 responses; an undecided comparison stops recovery without an estimate. Appendix~\ref{app:noise} specifies the assumptions and tests. We report exact recovery and mean token responses after pair discovery.

\begin{figure*}[t]
\centering\includegraphics[width=\textwidth]{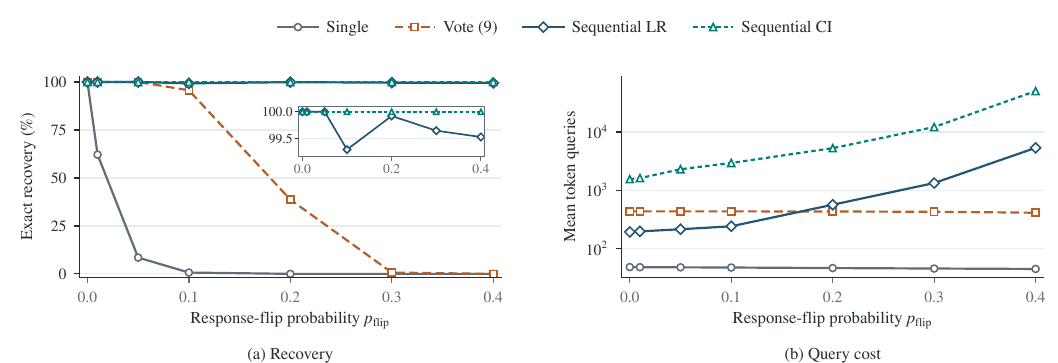}
\caption{Exact recovery and cost under independent response flips. Single uses one response per comparison; Majority-9 uses nine-response voting; LR and CI use sequential likelihood-ratio and confidence-interval tests, respectively. Each point aggregates 504 fixed in-range pairs and five noise seeds. Cost counts responses after pair discovery. The inset enlarges LR and CI near 100\%; the cost axis is logarithmic. Lines connect measured conditions.}
\label{fig:noise}
\end{figure*}

Figure~\ref{fig:noise} and Table~\ref{tab:noise} evaluate independent flips of responses from fixed in-range pairs. At 10\% flips, LR achieves higher recovery and lower mean cost than nine-response voting. LR uses a noise-rate upper bound; CI recovers all simulated cases at higher cost without that bound. The 2,520 attempts comprise 504 numerical cases and five noise seeds.

\begin{table}[t]
\centering\small
\caption{Recovery at 10\% independent response flips.}
\label{tab:noise}
\begin{tabular*}{\linewidth}{@{\extracolsep{\fill}}lrr@{}}\toprule
Rule & Exact recovery & Mean responses\\\midrule
Single & 17/2,520 & 47.8\\
Majority-9 & 2,411/2,520 & 437.4\\
LR & 2,502/2,520 & 242.7\\
CI & 2,520/2,520 & 2,944.2\\
\bottomrule\end{tabular*}
\par\smallskip
\begin{minipage}{\linewidth}\footnotesize
\textit{Notes.} Rules are defined in Section~\ref{sec:noise-design}.
Exact recovery requires both scores and priority; all 2,520 attempts per rule are counted.
Responses include repetitions after pair discovery. No attempt is unresolved at this noise level.
\end{minipage}
\end{table}

A separate condition selects uniformly between tokens with equal maximum rounded scores; strict comparisons remain exact. We compare CI with a rule that emulates a chosen fixed priority: return token 1 only when all 13 repeated responses are token 1, and token 0 otherwise. Its guarantee and score-pair recovery objective are given in Appendix~\ref{app:noise}.

Under uniform random token tie breaking, CI leaves all 2,520 cases unresolved. Emulating fixed priority instead recovers $2{,}516/2{,}520$ score pairs at mean cost 632.1. These rules estimate comparisons for fixed scores; they do not undo the score changes measured in the batch experiment.

\section{Defense}\label{sec:defense}
\subsection{Subtracting the Maximum before Bias Addition}
Subtract the maximum score before adding any bias (Figure~\ref{fig:overview}(d)): $A_i(q)=\fl(q_i-\max_jq_j)$, computed coordinatewise.

\begin{proposition}[Max subtraction before bias addition]\label{prop:prebias}
For every allowed input $x$, let the finite stored binary32 vectors satisfy the exact real relation $q'(x)=q(x)+c(x)\mathbf1$. Both services apply correctly rounded max subtraction $A$ with consistent signed-zero handling. Suppose subsequent observable responses depend on the original scores only through $A(q(x))$, under the same conditional response kernel given the public history, input, and controls, including tie breaking and internal randomness. Then every adaptive observer obtains identical transcript distributions for any finite query budget or almost-surely finite stopping rule.
\end{proposition}
\begin{proof}
For each input, $q'_i-\max_jq'_j=q_i-\max_jq_j$. Correct rounding therefore gives $A(q')=A(q)$ bit for bit. Couple the observer's randomness and the services' responses. Equal histories give the same next query and response distribution, so induction gives identical histories at every finite round. A common stopping rule yields identical stopped transcripts.
\end{proof}

The exact-shift condition applies to the stored vectors. If a common shift is rounded earlier, it can already change their relative scores. For example, binary32 $q=(1,1-2^{-24})$ centers to $(0,-2^{-24})$, whereas componentwise rounding of $q+2^{24}\mathbf1$ gives $(2^{24},2^{24})$, which centers to $(0,0)$. The prior rounding has changed the relative scores, so the proposition's exact-shift premise no longer holds.

When the maximum is subtracted \emph{after} bias addition, rounding has already mapped the candidates to potentially different vectors. A later operation that preserves the winning token retains the distinguishing response. Stable softmax implementations already subtract a maximum to avoid overflow~\cite{goldberg1991floating,blanchard2021softmax}; the relevant question is its position relative to the exposed control.

\subsection{Effect on Distinguishing Queries}\label{sec:defense-evaluation}
The suppression experiment compares the original processor with subtraction of the maximum score either before or after bias addition. We take eight cached positions from each of two FP32 cohorts, Pythia-14M and Qwen3-0.6B, round their complete score vectors to grids $2^{-8}$ and $2^{-6}$, and apply shifts $-0.25,-0.125,0,0.125,0.25$. This gives 160 conditions in 32 groups, with bitwise-identical complete max-centered and FP32 log-softmax vectors within each group. We freeze recovery predictions and replay the union of observed bias payloads across shifts.

Figure~\ref{fig:defense} compares cross-shift winner disagreements for the original pipeline (Original), max subtraction before bias addition (Pre), and max subtraction afterward (Post). Applying 9,742 union payloads to the five shifts in each group gives 48,710 condition--payload evaluations and 38,968 cross-shift comparisons. Pre removes all 2,578 original disagreements; Post retains them.

\begin{figure}[t]
\centering\includegraphics[width=\columnwidth]{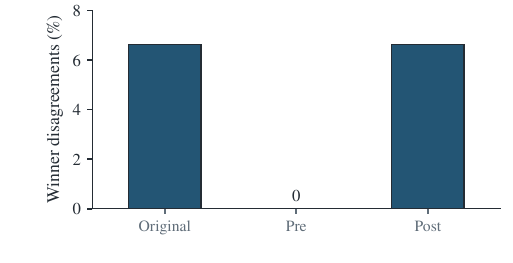}
\caption{Cross-shift winner disagreements under three processing orders. Pre and Post apply max subtraction before and after bias addition, respectively. Rates use 38,968 comparisons: 2,578 disagreements for Original and Post, and 0 for Pre. The experiment uses fixed-grid decoder groups; unperturbed probabilities and argmax decisions remain bitwise unchanged in all 160 conditions.}
\label{fig:defense}
\end{figure}

Preprocessing changes 10,038 of the 48,710 original \emph{biased} outputs. This is expected for payloads selected to expose fine rounding behavior: suppressing the channel changes those responses. In contrast, all 160 unperturbed argmax, softmax, and log-softmax outputs remain bitwise identical.

\subsection{Bias-Control Preservation}
To measure bias-control preservation, we use 16 evenly spaced inputs from each of the six serving configurations: 96 original cached vectors. Each receives 96 random sparse biases, with eight tokens from the full vocabulary or top 64 and 16 draws per support and amplitude $0.5,2,10$. It also receives 33 target-only biases $\fl(d+kU/8)$, $k=-16,\ldots,16$, where $d$ is the top-two gap and $U$ is the binary32 spacing at their maximum absolute score. We compare processing orders, complete FP32 probabilities, and two fixed offline tie-breaking rules: smallest or largest token ID. All payloads satisfy $H=100$. The inputs, seed, and bias-generation rules are fixed before execution; Appendix~\ref{app:repro} gives the details.

Pre-bias max subtraction preserves all 96 unperturbed greedy decisions and complete probability vectors bit for bit. It changes no greedy decision across 9,216 random biases; the mean complete-probability $L_1$ difference is $6.08\times10^{-8}$ and its maximum is $9.41\times10^{-7}$. The biases themselves change the unperturbed winner in 1,776 cases, including 1,737 of 4,608 top-64 controls.

Among 3,168 biases near the comparison threshold, preprocessing changes 738 maximizing sets and 369 greedy decisions under each fixed tie-breaking rule. Every change separates an original rounded tie; no strict-winner reversal occurs. Post-bias subtraction preserves all decisions and complete probabilities. Independent NumPy replay reproduces every FP32 processing vector. Thus, fine control near a decision boundary can change token selection even when the probability difference is small.

Processing order also matters for other controls. Repetition penalties can depend on the common additive component, so moving max subtraction across such a penalty can change its effect.

\section{Related Work}
\subsection{Prediction-Interface and Language Model Extraction}
Tram\`er et al.\ established extraction from prediction interfaces~\cite{tramer2016stealing}; subsequent work studies high-fidelity extraction and its relationship to recovered parameters~\cite{jagielski2020highaccuracy,carlini2020cryptanalytic}. 

For language models, Morris et al.\ use biases and token responses to recover probability information~\cite{morris2023inversion}. Carlini et al.\ recover output-projection information and develop efficient logprob-free score-difference attacks~\cite{carlini2024stealing}. Finlayson et al.\ analyze the low-dimensional linear structure exposed by API logits~\cite{finlayson2024logits}. The softmax bottleneck is established~\cite{yang2018softmax}. Our contribution concerns the common additive component left unresolved by ideal normalized or biased-argmax observations.

Geometric output signatures use constraints shared across model inputs. Finlayson et al.\ exploit an ellipse associated with the output projection for source verification~\cite{finlayson2025signature}. We study a per-input observation that recovers original logit values.

\subsection{Finite Precision and Numerical Side Channels}
Numerical analysis studies rounding and evaluation order~\cite{goldberg1991floating,higham2002accuracy,ieee2019floating}, including stable softmax~\cite{blanchard2021softmax}. Andrysco et al.\ demonstrate operand-dependent floating-point timing leakage, including subnormal arithmetic~\cite{andrysco2015subnormal}. We instead observe token selection under chosen biases near a normal exponent boundary.

Machine-learning implementations also expose execution-dependent signatures. iNNformant uses numerical differences in inference for device identification~\cite{innformant2021}, and HSPI studies hardware signatures in model outputs~\cite{hspi2025}. Our setting fixes the implementation and recovers a score component through chosen decoder controls.

Quantamination exploits shared dynamic-quantization parameters to leak data across a batch~\cite{foerster2026quantamination}, while FloatDoor modifies a model to turn platform-dependent numerical differences into a backdoor trigger~\cite{loose2026floatdoor}. Our mechanism requires neither a co-batched victim input nor modified model weights; it recovers the numerical values entering the exposed bias operation.

Kadav et al.\ study softmax-preserving output-head reparameterizations to improve output-head quantization~\cite{kadav2026softmax}. Their work establishes that ideal-equivalent representatives can behave differently after quantization. Our direction is observational: determine which representative-dependent information can be identified externally from token responses, and provide a local recovery construction.

\subsection{Decoder Controls and Interface Semantics}
Decoding controls operate on model logits. Repetition penalties originate in work such as CTRL~\cite{keskar2019ctrl}; Baichuan 2 discusses their dependence on logit values beyond relative differences~\cite{baichuan2023baichuan2}. This dependence on the common additive component is therefore established. Our recovery construction instead concerns the rounding of an additive bias, whose ideal real-arithmetic counterpart preserves the common-additive-shift ambiguity.

Transformers and vLLM expose implementations of logit processing and token selection~\cite{wolf2020transformers,vllm2026sampler}. We analyze the rounded addition used by such processors and evaluate recovery end to end in LLM serving systems, linking the numerical operator to observable token responses.

\section{Discussion}
\subsection{Range Identification and Future Work}\label{sec:range-discussion}
Opaque endpoints require checks of bias resolution, operation order, score stability, and reachable ranges. Controlled passive-equality experiments isolate the mechanism; its prevalence in deployments requires separate measurement.

Range selection affects both coverage and cost. Outside the calibrated intervals, erroneous candidates often preserve the score gap while shifting both logits toward a window boundary. The seven single-window false acceptances in Section~\ref{sec:native-evaluation} illustrate this ambiguity; the second window rejects those candidates by providing a contradictory response. Globally justified range tests and recovery across binades remain open problems under bounded biases.

Extending complete output-head extraction to larger vocabularies requires measuring the relative-coordinate cost, which dominates the present study. Variable execution conditions require a target definition that accounts for changing scores; independent response-noise guarantees apply to a fixed target.

\subsection{Research Ethics and Reproducibility}
We use public checkpoints, synthetic inputs, research-controlled CPU/GPU resources, and one output token per bias setting; the output-head experiment batches independent settings. Controlled variants modify locally held vectors. All fixed cases remain in the records, including unmet preconditions. The study collects no other users' data and involves no access-control bypass.

Our experimental records include input generators, pinned model revisions, query transcripts, frozen predictions, and CPU replay checks (Appendix~\ref{app:repro}).

We are preparing disclosure materials for the vLLM, SGLang, and llama.cpp maintainers, including minimal reproductions, tested configurations, and the pre-bias centering analysis. We plan to coordinate validation and mitigation with the maintainers.

\noindent\textbf{Use of generative AI.}
We used ChatGPT and Codex to assist literature review, theoretical development, experimental design, implementation, data analysis, and manuscript drafting, revision, and translation. The authors are responsible for reviewing and validating AI-assisted material and for the accuracy and originality of the final work.

\section{Conclusion}
Floating-point bias addition can disclose the common additive component of logits that ideal probabilities and ideal bias comparisons leave unidentified. We characterize this difference, provide a local binary32 recovery method with explicit query bounds, and verify exact recovery through complete generation calls. A token-only output-head experiment verifies the recovery of the shared row component under an invertible change of hidden basis. Preprocessing-order experiments identify a condition that suppresses disclosure. The results show that the numerical implementation of a token-only interface is part of its information-disclosure behavior.

\bibliographystyle{IEEEtran}
\bibliography{references}

@inproceedings{morris2023inversion,
 author={Morris, John X. and Zhao, Wenting and Chiu, Justin T. and Shmatikov, Vitaly and Rush, Alexander M.},
 title={Language Model Inversion},
 booktitle={International Conference on Learning Representations}, year={2024}}

@inproceedings{carlini2024stealing,
 author={Carlini, Nicholas and Paleka, Daniel and Dvijotham, Krishnamurthy Dj and Steinke, Thomas and Hayase, Jonathan and Cooper, A. Feder and Lee, Katherine and Jagielski, Matthew and Nasr, Milad and Conmy, Arthur and Wallace, Eric and Rolnick, David and Tram{\`e}r, Florian},
 title={Stealing Part of a Production Language Model},
 booktitle={Proceedings of the 41st International Conference on Machine Learning},
 series={Proceedings of Machine Learning Research}, volume={235}, pages={5680--5705}, year={2024}}

@inproceedings{finlayson2024logits,
 author={Finlayson, Matthew and Ren, Xiang and Swayamdipta, Swabha},
 title={Logits of {API}-Protected {LLMs} Leak Proprietary Information},
 booktitle={Conference on Language Modeling}, year={2024}}

@inproceedings{finlayson2025signature,
 author={Finlayson, Matthew and Ren, Xiang and Swayamdipta, Swabha},
 title={Every Language Model Has a Forgery-Resistant Signature},
 booktitle={International Conference on Learning Representations}, year={2026}}

@misc{kadav2026softmax,
 author={Kadav, Asim and Flores, Christian and Arora, Chirag and Kotte, Varun and Zheng, Hongbo and Yan, Lan and Shanmugasundaram, Priya and King, Tracy Holloway},
 title={Softmax Reparameterization for Output-Head Quantization},
 year={2026}, howpublished={arXiv:2609.31291}, note={Version 2}}

@misc{keskar2019ctrl,
 author={Keskar, Nitish Shirish and McCann, Bryan and Varshney, Lav R. and Xiong, Caiming and Socher, Richard},
 title={{CTRL}: A Conditional Transformer Language Model for Controllable Generation},
 year={2019}, howpublished={arXiv:1909.05858}}

@misc{baichuan2023baichuan2,
 author={Yang, Aiyuan and Xiao, Bin and Wang, Bingning and others},
 title={Baichuan 2: Open Large-scale Language Models},
 year={2023}, howpublished={arXiv:2309.10305}}

@inproceedings{innformant2021,
 author={Schl{\"o}gl, Alexander and Kupek, Tobias and B{\"o}hme, Rainer},
 title={{iNNformant}: Boundary Samples as Telltale Watermarks},
 booktitle={Proceedings of the 2021 ACM Workshop on Information Hiding and Multimedia Security},
 pages={81--86}, publisher={ACM}, year={2021},
 doi={10.1145/3437880.3460411},
 url={https://doi.org/10.1145/3437880.3460411}}

@inproceedings{hspi2025,
 author={Zhang, Cheng and Foerster, Hanna and Mullins, Robert D. and Zhao, Yiren and Shumailov, Ilia},
 title={Hardware and Software Platform Inference},
 booktitle={Proceedings of the 42nd International Conference on Machine Learning},
 series={Proceedings of Machine Learning Research}, volume={267}, pages={74882--74895}, year={2025}}

@misc{foerster2026quantamination,
 author={Foerster, Hanna and Shumailov, Ilia and Zhang, Cheng and Zhao, Yiren and Hayes, Jamie and Mullins, Robert},
 title={Quantamination: Dynamic Quantization Leaks Your Data Across the Batch},
 year={2026}, howpublished={arXiv:2604.26505}}

@misc{loose2026floatdoor,
 author={Loose, Nils and Sander, Jonas and M{\"a}chtle, Felix and Eisenbarth, Thomas},
 title={{FloatDoor}: Platform-Triggered Backdoors in {LLMs}},
 year={2026}, howpublished={arXiv:2606.19535}}

@article{goldberg1991floating,
 author={Goldberg, David},
 title={What Every Computer Scientist Should Know About Floating-Point Arithmetic},
 journal={ACM Computing Surveys}, volume={23}, number={1}, pages={5--48}, year={1991},
 doi={10.1145/103162.103163}}

@misc{ieee2019floating,
 author={{IEEE}}, title={{IEEE} Standard for Floating-Point Arithmetic},
 howpublished={IEEE Std 754-2019}, year={2019},
 doi={10.1109/IEEESTD.2019.8766229},
 url={https://standards.ieee.org/ieee/754/6210/}}

@article{blanchard2021softmax,
 author={Blanchard, Pierre and Higham, Desmond J. and Higham, Nicholas J.},
 title={Accurately Computing the Log-Sum-Exp and Softmax Functions},
 journal={IMA Journal of Numerical Analysis}, volume={41}, number={4}, pages={2311--2330}, year={2021},
 doi={10.1093/imanum/draa038}}

@inproceedings{tramer2016stealing,
 author={Tram{\`e}r, Florian and Zhang, Fan and Juels, Ari and Reiter, Michael K. and Ristenpart, Thomas},
 title={Stealing Machine Learning Models via Prediction {APIs}},
 booktitle={25th USENIX Security Symposium}, year={2016}, pages={601--618}}

@inproceedings{biderman2023pythia,
 author={Biderman, Stella and Schoelkopf, Hailey and Anthony, Quentin Gregory and Bradley, Herbie and O'Brien, Kyle and Hallahan, Eric and Khan, Mohammad Aflah and Purohit, Shivanshu and Prashanth, Usvsn Sai and Raff, Edward and Skowron, Aviya and Sutawika, Lintang and Van Der Wal, Oskar},
 title={Pythia: A Suite for Analyzing Large Language Models Across Training and Scaling},
 booktitle={Proceedings of the 40th International Conference on Machine Learning},
 series={Proceedings of Machine Learning Research}, volume={202}, pages={2397--2430}, year={2023}}

@misc{qwen2025report,
 author={Yang, An and Li, Anfeng and Yang, Baosong and others}, title={{Qwen3} Technical Report},
 howpublished={arXiv:2505.09388}, year={2025}}

@inproceedings{wolf2020transformers,
 author={Wolf, Thomas and Debut, Lysandre and Sanh, Victor and others},
 title={Transformers: State-of-the-Art Natural Language Processing},
 booktitle={Proceedings of the 2020 Conference on Empirical Methods in Natural Language Processing: System Demonstrations},
 year={2020}, pages={38--45}, doi={10.18653/v1/2020.emnlp-demos.6}}

@misc{huggingface2026generation,
 author={{Hugging Face}}, title={Utilities for Generation},
 year={2026}, howpublished={Transformers 4.57.1 documentation}, note={Accessed 2026-10-02},
 url={https://huggingface.co/docs/transformers/v4.57.1/en/internal/generation_utils}}

@misc{openrouter2026parameters,
 author={{OpenRouter}}, title={{API} Parameters},
 year={2026}, note={Accessed 2026-10-01},
 url={https://openrouter.ai/docs/api_reference/parameters}}

@misc{vllm2026sampler,
 author={{vLLM Contributors}}, title={{vLLM} {V1} {GPU} Sampling Implementation},
 year={2026}, note={Commit cb4f016c5d02c9d483e6e3f1702d51da239b1a9b},
 url={https://github.com/vllm-project/vllm/blob/cb4f016c5d02c9d483e6e3f1702d51da239b1a9b/vllm/v1/worker/gpu/sample/sampler.py}}

@book{higham2002accuracy,
 author={Higham, Nicholas J.}, title={Accuracy and Stability of Numerical Algorithms},
 edition={2}, publisher={Society for Industrial and Applied Mathematics}, year={2002},
 doi={10.1137/1.9780898718027}}

@inproceedings{yang2018softmax,
 author={Yang, Zhilin and Dai, Zihang and Salakhutdinov, Ruslan and Cohen, William W.},
 title={Breaking the Softmax Bottleneck: A High-Rank {RNN} Language Model},
 booktitle={International Conference on Learning Representations}, year={2018}}

@inproceedings{jagielski2020highaccuracy,
 author={Jagielski, Matthew and Carlini, Nicholas and Berthelot, David and Kurakin, Alex and Papernot, Nicolas},
 title={High Accuracy and High Fidelity Extraction of Neural Networks},
 booktitle={29th USENIX Security Symposium}, year={2020}, pages={1345--1362},
 publisher={USENIX Association}}

@inproceedings{carlini2020cryptanalytic,
 author={Carlini, Nicholas and Jagielski, Matthew and Mironov, Ilya},
 title={Cryptanalytic Extraction of Neural Network Models},
 booktitle={Advances in Cryptology---CRYPTO 2020},
 series={Lecture Notes in Computer Science}, volume={12172},
 year={2020}, pages={189--218}, publisher={Springer},
 doi={10.1007/978-3-030-56877-1_7}}

@misc{radford2019gpt2,
 author={Radford, Alec and Wu, Jeffrey and Child, Rewon and Luan, David and Amodei, Dario and Sutskever, Ilya},
 title={Language Models are Unsupervised Multitask Learners},
 year={2019}, howpublished={OpenAI technical report},
 url={https://cdn.openai.com/better-language-models/language_models_are_unsupervised_multitask_learners.pdf}}

@misc{black2021gptneo,
 author={Black, Sid and Gao, Leo and Wang, Phil and Leahy, Connor and Biderman, Stella},
 title={{GPT-Neo}: Large Scale Autoregressive Language Modeling with Mesh-Tensorflow},
 year={2021}, howpublished={Zenodo}, doi={10.5281/zenodo.5297715},
 url={https://doi.org/10.5281/zenodo.5297715}}

@misc{zhang2022opt,
 author={Zhang, Susan and Roller, Stephen and Goyal, Naman and others},
 title={{OPT}: Open Pre-trained Transformer Language Models},
 year={2022}, howpublished={arXiv:2205.01068}}

@misc{bigscience2022bloom,
 author={{BigScience Workshop}},
 title={{BLOOM}: A 176B-Parameter Open-Access Multilingual Language Model},
 year={2022}, howpublished={arXiv:2211.05100}}

@inproceedings{benallal2025smollm2,
 author={Ben Allal, Loubna and Lozhkov, Anton and Bakouch, Elie and others},
 title={{SmolLM2}: When Smol Goes Big---Data-Centric Training of a Fully Open Small Language Model},
 booktitle={Conference on Language Modeling}, year={2025},
 url={https://openreview.net/forum?id=3JiCl2A14H}}

@inproceedings{andrysco2015subnormal,
 author={Andrysco, Marc and Kohlbrenner, David and Mowery, Keaton and Jhala, Ranjit and Lerner, Sorin and Shacham, Hovav},
 title={On Subnormal Floating Point and Abnormal Timing},
 booktitle={2015 IEEE Symposium on Security and Privacy}, year={2015},
 pages={623--639}}

@misc{kalamkar2019bfloat16,
 author={Kalamkar, Dhiraj and Mudigere, Dheevatsa and Mellempudi, Naveen and Das, Dipankar and Banerjee, Kunal and Avancha, Sasikanth and Vooturi, Dharma Teja and Jammalamadaka, Nataraj and Huang, Jianyu and Yuen, Hector and Yang, Jiyan and Park, Jongsoo and Heinecke, Alexander and Georganas, Evangelos and Srinivasan, Sudarshan and Kundu, Abhisek and Smelyanskiy, Misha and Kaul, Bharat and Dubey, Pradeep},
 title={A Study of {BFLOAT16} for Deep Learning Training},
 year={2019}, howpublished={arXiv:1905.12322}}

@article{wald1945sequential,
 author={Wald, Abraham},
 title={Sequential Tests of Statistical Hypotheses},
 journal={The Annals of Mathematical Statistics}, volume={16}, number={2},
 pages={117--186}, year={1945}, doi={10.1214/aoms/1177731118}}

@article{howard2021timeuniform,
 author={Howard, Steven R. and Ramdas, Aaditya and McAuliffe, Jon and Sekhon, Jasjeet},
 title={Time-uniform, nonparametric, nonasymptotic confidence sequences},
 journal={The Annals of Statistics}, volume={49}, number={2},
 pages={1055--1080}, year={2021}, doi={10.1214/20-AOS1991}}

@misc{he2025nondeterminism,
 author={He, Horace}, title={Defeating Nondeterminism in {LLM} Inference},
 howpublished={Thinking Machines Lab}, year={2025}, month={September},
 url={https://thinkingmachines.ai/blog/defeating-nondeterminism-in-llm-inference/}}

@inproceedings{vaswani2017attention,
 author={Vaswani, Ashish and Shazeer, Noam and Parmar, Niki and Uszkoreit, Jakob and Jones, Llion and Gomez, Aidan N. and Kaiser, {\L}ukasz and Polosukhin, Illia},
 title={Attention Is All You Need},
 booktitle={Advances in Neural Information Processing Systems}, volume={30}, year={2017}}

@inproceedings{kwon2023pagedattention,
 author={Kwon, Woosuk and Li, Zhuohan and Zhuang, Siyuan and others},
 title={Efficient Memory Management for Large Language Model Serving with {PagedAttention}},
 booktitle={Proceedings of the 29th Symposium on Operating Systems Principles},
 year={2023}, pages={611--626}, doi={10.1145/3600006.3613165}}

@inproceedings{zheng2024sglang,
 author={Zheng, Lianmin and Yin, Liangsheng and Xie, Zhiqiang and others},
 title={{SGLang}: Efficient Execution of Structured Language Model Programs},
 booktitle={Advances in Neural Information Processing Systems}, volume={37}, year={2024}}

@misc{llamacpp2026server,
 author={{ggml-org Contributors}}, title={llama.cpp Server}, year={2026},
 note={Commit 7f2dd88b0ac393357ae6a9e1992185a48c20b13b},
 url={https://github.com/ggml-org/llama.cpp/blob/7f2dd88b0ac393357ae6a9e1992185a48c20b13b/tools/server/README.md}}

@misc{karpathytinyllamas,
  author = {Andrej Karpathy},
  title = {{TinyLlamas: TinyStories-260K Checkpoint}},
  howpublished = {Hugging Face model repository},
  url = {https://huggingface.co/karpathy/tinyllamas},
  note = {Revision 0bd21da7698eaf29a0d7de3992de8a46ef624add, stories260K; accessed October 9, 2026}
}

\appendices
\section{Boundary Cases in Local Recovery}\label{app:boundary}
The recovery proof uses independently formed payloads. If a processor instead evaluates $\fl(\fl(q_i+B)+d)$, it can have a different rounding history from $\fl(q_i+\fl(B+d))$. The artifact checks the latter operation order and rejects payloads that are not exactly representable or exceed the allowed budget. This distinction matters most near the exponent boundary used by the probes.

For the tie-breaking priority test, $p=\max\{\ell,F+d_0+\Delta\}$ satisfies $p\le a$ under both hypotheses because $t>F$ and the scores lie on the $\Delta$ grid. The last pivot is at most $u+2\Delta\le P$. Its anchor bias is therefore nonnegative, and the first target bias is at most $H_0$; later target biases are smaller. When the true scores are equal, the original-anchor definition enforces anchor tie-breaking priority and the same test applies.

For residue identification, the two responses at $p_0,p_0+\Delta$ are $00,10,11,01$ for residues $0,1,2,3$, apart from $k=0$, which also produces $01$. The query at $p_0+3\Delta$ produces one for $k=0$ and zero for nonnegative residue-three candidates. If $p_0<u$, then $p_0\le u-\Delta\le P-3\Delta$, so that extra pivot is valid. If $p_0=u$, the algorithm returns before making it.

The lower-bound state count includes the unknown pairwise tie-breaking priority. For each anchor index $i$, the $i-1$ strictly lower targets admit two priorities, while the equal target admits only the original-anchor priority. Summing $2(i-1)+1$ over $i=1,\ldots,K-2$ gives $(K-2)^2$. If only the score pair must be returned and tie-breaking priority is a nuisance variable, the number of distinct pairs is $(K-2)(K-1)/2$, giving a 45-query counting lower bound for binary32. The 46-query comparison in the main text therefore explicitly includes tie-breaking priority. Allowing nonbinary responses also changes the argument.

\section{A Bounded-Bias Nonidentifiability Region}
\label{app:deadband}
\begin{proposition}[An explicit BF16 nonidentifiability region]
\label{prop:deadband}
Let $q$ be a finite BF16 vector, let $M=\max_iq_i\in\{360,362,\ldots,410\}$, and define $q'_i=\operatorname{RN}_{\mathrm{BF16}}(q_i-2)$. Under~\eqref{eq:oracle}, allowing arbitrary simultaneously biased tokens with every payload in $[-100,100]$, the complete sets of maximizing indices are equal for $q$ and $q'$ for every query. They therefore have identical adaptive winner transcripts under the same fixed tie-breaking order.
\end{proposition}
\begin{proof}
Let $S=\{i:q_i\ge M-200\}$. For $i\in S$, $160\le q_i\le410$ and subtracting two is exact in BF16, including across $256$. Hence $q'_i=q_i-2$. Every other score is at most $M-201$, and its shifted-and-rounded value is at most $M-203$. With biases bounded by $100$, such a token remains strictly below a maximum-score token in both worlds and cannot win.

For $i\in S$, put $\beta_i=\fl(b_i)$ and $x_i=q_i+\beta_i$, where the latter denotes the exact real sum. Its maximum $x_*$ lies in $[M-100,M+100]\subset[260,510]$. The translated maximum lies in $[258,508]$. Both maxima and their rounding cells are strictly inside the binary32 binade $[256,512)$, with spacing $\delta=2^{-15}$. Translation by $-2=-65536\delta$ shifts each relevant rounding cell by an even number of grid steps, preserving ties-to-even endpoint membership. Therefore
\[
 \begin{gathered}
 \fl(x_i)=\fl(x_*)\\
 \iff\quad \fl(x_i-2)=\fl(x_*-2).
 \end{gathered}
\]
This equivalence also holds for $x_i$ outside the maximum's binade because membership is in the translated maximum rounding cell. Thus the complete maximizing sets agree. The same tie-breaking order gives equal responses to every query; induction on the resulting histories gives identical adaptive transcripts. Rounding the most negative finite BF16 value after subtracting two remains finite, so $q'$ is a valid score vector.
\end{proof}

\section{Negative-Binade Construction}
\label{app:negative}
The sign of the scores affects which side of an exponent boundary has finer
spacing. This distinction permits a separate construction for negative model
scores; no common score shift is inserted into the service.

Let the addition format have binary precision $p\ge3$, and let $S=2^e$ be public.
Set $\delta=S/2^{p-1}$, $K=\lfloor\min(H,S)/\delta\rfloor$,
$H_0=K\delta$, and $F=-S-H_0$. Assume $K\ge2$, that the original greedy winner has
score $a$, and that the selected target has score $t$, with
\[
 a\ge t>F,\qquad a\in[\ell,u]\subset(F,-S-\delta].
\]
The endpoints are on the $\delta$ lattice. Both scores therefore lie in the
negative binade $(-2S,-S)$. The interface uses separately cast biases, a single
correctly rounded addition, and an unknown fixed global tie-breaking order. All constructed
biases are nonnegative and at most $H_0$. The format must represent $\delta/2$;
we exclude overflow and flushing of the required subnormal operands.

The following proves Theorem~\ref{thm:negative} and its precision-dependent extension.

\begin{proof}
Write $d=a-t$. Adding $j\delta$ to the target gives a monotone threshold because
all score sums near the comparison remain exactly on the original lattice. Let
$j$ be the first winning integer. Binary search costs at most
$\lceil\log_2K\rceil$ queries and leaves two hypotheses:
$d=d_0=(j-1)\delta$ with anchor tie-breaking priority, or $d=d_0+\delta$ with target tie-breaking priority. The search endpoints follow from the original winner and
$t+H_0>-S>a$.

Set $r=F+d_0+\delta$ and submit biases
\[
 b_a=-S-r,\qquad b_t=-S-r+d_0+\delta/2.
\]
These are legal: $b_t=H_0-\delta/2$ and $0\le b_a\le H_0$.
Under anchor priority, $a-r=t-F-\delta\ge0$, so the two exact sums are $x$ and
$x+\delta/2$, with $x=-S+(a-r)\ge-S$. Both are representable in the finer region,
and the target wins. Under target priority, $a-r=t-F\ge\delta$, and the exact
sums are $x$ and $x-\delta/2$, both strictly in the finer region; the anchor wins.
Thus one query determines the true gap and tie-breaking priority.

For a pivot $r$ in the remaining legal range, first align the two exact sums to
$x=-S+(a-r)$ using biases $-S-r$ and $-S-r+d$. Increase the bias of the
lower-priority token by $\delta/2$. With $k=(a-r)/\delta$, its winning indicator is
\[
 g(k)=\begin{cases}
 1,&k\ge0,\\
 1,&k<0\text{ and }k\text{ odd},\\
 0,&k<0\text{ and }k\text{ even}.
 \end{cases}
\]
For $k\ge0$, the half-step is exactly representable. For $k<0$, it is a midpoint
in the coarser binade; ties-to-even rounds upward precisely for odd $k$.
Choose the initial residue-test pivot $r=u+\delta$. It satisfies $r>a$ and $r\le-S$,
so the response identifies the parity of $(a-r)/\delta$ without a boundary
ambiguity. Restrict subsequent candidate pivots to the inferred parity class.
Then $k$ is even, and the response is exactly $\mathbf 1[a\ge r]$. Binary search
on at most $\lceil N/2\rceil$ candidates completes recovery.

Every parity-test or search pivot is at least $F+d+\delta$. Thus its target payload,
even after the half-step increment, is at most $H_0-\delta/2$; the other payload
is nonnegative and no larger. All payloads are exact multiples of $\delta/2$
within $[0,S]$, hence representable. The anchor is never decreased, so no
unmodified token can displace the original winner under the fixed global tie
order. Each response is consequently binary.

Since $K\le2^{p-1}$ and $N\le K-1$, the upper bound is $2p-1$. For the lower
bound, take $H_0=S$ and scores $a=F+i\delta$, $t=F+j\delta$ with
$1\le j\le i\le K-1$. There are $(K-1)^2$ allowed score/priority triples. A binary
decision tree requires at least $\lceil\log_2(K-1)^2\rceil=2p-2$ queries to
identify all of them. This count concerns triple recovery, not score-only
recovery or an oracle with more than two possible responses.
\end{proof}

\noindent\textbf{Other known rounding directions.}
The finer-region gap-disambiguation query is exact and does not depend on the
rounding direction. The final comparison is adapted to the specified rounding direction. Under rounding toward negative
infinity, or nearest with ties away from zero, the lower-priority half-step
increment wins exactly when $a\ge r$. Under rounding toward positive infinity,
or toward zero (all constructed sums are negative), instead decrease the
higher-priority token's aligned bias by $\delta/2$. The lower-priority token then
wins exactly when $a>r$. The latter probe is legal for search pivots
$r\le u\le-S-\delta$. Ordinary binary search gives the same $2p-1$ uniform
bound. These are algorithms for a \emph{known} rounding direction, not an
identification theorem for an unknown runtime mode. Nearest-away is analyzed as
an IEEE-defined direction; its support is not assumed for binary hardware.

\section{Storage Precision and Control Resolution}
\label{app:mixed}
Stored-score precision can also reduce the candidate set. Let $s$ be the significand precision of a known score-storage format and let $p\ge s+2$ be the bias and addition precision. The scores lie in $g\mathbb Z$, where $g=P2^{-s}$ below a positive boundary $P$ and $g=S2^{1-s}$ below a negative boundary $-S$. The appropriate local range and representability conditions still apply. Storage precision and arithmetic precision are distinct interface properties~\cite{kalamkar2019bfloat16}.
\begin{proposition}[Recovery with known score precision]\label{prop:mixed}
Let $B=P$ or $B=-S$ be the chosen boundary and $g$ the corresponding storage spacing defined above. Set $H_0=Kg\le H$ and $F=B-H_0$, where $K$ is an integer and $2\le K\le2^{s-1}$. Suppose $a,t\in g\mathbb Z$, $a\ge t>F$, and $a\in[\ell,u]\subset(F,B-g]$, with endpoints on the storage grid. Assume the required bias increments are representable and the arithmetic neither overflows nor flushes required subnormals. With the original anchor, unknown fixed tie-breaking priority, and round-to-nearest, ties-to-even addition, exact pair-and-priority recovery requires at most
\begin{equation}
 \lceil\log_2K\rceil+1+\lceil\log_2N\rceil\le2s-1
 \label{eq:mixed-cost}
\end{equation}
binary recovery queries, where $N$ is the number of remaining anchor grid points after recovering the gap. BF16-to-binary32 and binary16-to-binary32 interfaces admit bounds of 15 and 21, respectively.
\end{proposition}

\begin{proof}
Search the target-bias threshold on the coarse storage grid. If the first winning
bias is $jg$, the remaining hypotheses are $d=(j-1)g$ with anchor tie-breaking priority
and $d=jg$ with target tie-breaking priority. A single bias $(j-1/2)g$ on the target gives
score $a+g/2$ under the former and $a-g/2$ under the latter. Both are exactly
representable in the arithmetic format, so this query identifies the gap and
tie-breaking priority.

For a positive boundary $P$, let $\Delta=P/2^p$ be the arithmetic spacing below
$P$. The storage spacing satisfies $g=2^{p-s}\Delta$. Choose pivots on the
storage grid, align the scores at $P+(a-r)$, and increase the lower-priority
token's bias by $\Delta$. Since $(a-r)/\Delta$ is a multiple of four, the
nearest-even response rule reduces directly to $\mathbf1[a<r]$; no residue-identification
queries are needed. For a negative boundary $-S$, let
$\delta=S/2^{p-1}$. Align at $-S+(a-r)$ and use an increment $\delta/2$ instead.
Here $(a-r)/\delta$ is even, so the winning indicator is
$\mathbf1[a\ge r]$. Each case yields ordinary binary search on the storage
candidates.

After finding the gap, intersect the public interval with
$[F+d+g,\,\text{boundary}-g]$. All search biases are nonnegative and bounded by
$H_0$, including the small increment, because $g$ exceeds that increment. They
are exactly representable in the finer arithmetic format. Monotonicity and the
fixed global tie-breaking order exclude untouched tokens, as in the binary32
construction. The gap and anchor searches each cost at most $s-1$ queries.
\end{proof}

The lattice condition applies to scores entering bias addition; advertised weight or forward precision alone does not establish it. The proposition also requires biases to retain the finer arithmetic precision.

\noindent\textbf{A finite-alphabet control.}
For $q=(30,28)$ and $q'=q-\tfrac18\mathbf1$, fine binary32 biases produce a
winner difference. When biases are instead restricted to all finite BF16 values
in $[-100,100]$, and addition remains binary32, the two pair oracles are
identical for both tie-breaking priorities. The same holds for the binary16 bias alphabet.
We verified this finite statement by comparing the first-winning target-bias
index for every anchor bias: 34,193 BF16 values and 44,161 binary16 values.
Monotonicity extends the threshold comparison to every bias pair in each specified Cartesian alphabet for these two vectors. If both the bias
and addition precision are lowered, format-scaled distinguishing probes exist
again.

\noindent\textbf{A sufficient condition for exact addition.}
Let $g$ be a representable power of two and $p$ the addition precision. Suppose every hidden score and accepted bias is a multiple of $g$, and $q'=q+c\mathbf1$ with $c\in g\mathbb Z$. If every possible biased sum in both vectors has magnitude at most $2^p g$, with no overflow, each sum is an exactly representable $ng$ with $|n|\le2^p$. All additions are therefore exact. The biased vectors differ by $c\mathbf1$, so their maximizer sets and winners coincide under the same fixed tie-breaking order. Identical responses at each round give identical adaptive transcripts.

This representability condition identifies when control quantization removes the nonlinearity used by the channel. Coarse biases alone are insufficient: with
binary32 scores $q=(2^{24},2^{24}-2)$, $q'=q-2\mathbf1$, and integer biases
$(0,3)$, the first vector yields equal rounded scores; the second has a strict target
winner. Here the exact-lattice magnitude condition is violated.

\section{Adaptive Selection Protocol}\label{app:anchor-selection}
For Table~\ref{tab:evidence}, enumerate $B=\pm2^e$, $e=-6,\ldots,10$, in ascending exponent order, positive before negative. Use each full local interval and require candidates to reproduce the pair's combined transcript. Set $\eta=|B|/2^{24}$ for $B>0$ and $\eta=|B|/2^{23}$ otherwise. Probe each candidate with
\[
 b_0=\fl(B-\hat a+s\eta),\quad
 b_1=\fl(B-\hat t+(s+u)\eta).
\]
Enumerate $s=0,\pm1,\pm2,\pm3,\pm4$ and $u=0,\pm1/4,\pm1/2,\pm3/4,\pm1$, negative before positive within each pair. Retain only legal payloads with $b_0\ge0$ and reject at the first contradiction. At most 81 checks per candidate precede the unique-candidate rule and eight-pair limit in Section~\ref{sec:anchor-selection}.

The two-window controller uses $P\in\{16,32\}$, $\eta=P/2^{24}$, $s\in\{0,-1,1\}$, and $u\in\{-3/4,3/4\}$. The cached ablations freeze predictions before scoring; the two BF16 controllers return identical selected pairs, scores, and priorities. Their full replay counts and input hashes are retained in the artifact.

\noindent\textbf{Complete query costs.} The 34-window generation controller uses 155,390 calls on 64 GPT-Neo inputs; its median/maximum is 2,245/3,927. On 32 BF16 OPT inputs, generic and BF16-aware recovery use 40,486 and 14,048 calls, with medians 1,184.5 and 413. Generic Qwen3 recovery uses 33,915 calls on 32 inputs, with median 994.5. Native Qwen3 totals are 15,312/15,417/14,768 for vLLM/SGLang/llama.cpp across 128 inputs each. Mean costs decompose into 2.2/2.2/2.2 selection calls, 111.4/112.2/107.0 window-search calls, and 6.0/6.0/6.2 checks. Repeated payloads and retries are included; calibration and independent replay are excluded.

\noindent\textbf{Single-window ablation.} With eight-pair selection, $H=100$, and unchanged checks, $P=16$ yields 64/64/60 Qwen3 recoveries across vLLM/SGLang/llama.cpp; $P=32$ yields 37/36/36. SmolLM2 yields 1/64 and 64/64 per framework, respectively. Qwen3 on llama.cpp returns four incorrect candidates at $P=16$ and three at $P=32$; other failures are unresolved. All seven pass single-window checks outside the local domain but contradict the two-window transcript for the same pair and mask. All 163,242 unique queries match full-vocabulary replay.

\section{Recovery with Noisy Binary Responses}\label{app:noise}
These results condition on a correctly isolated token pair, a fixed score vector in the public interval, and the original arithmetic assumptions. Let $v\in\{0,1\}$ be the ideal answer to a logical query. The controller uses at most $M$ such queries on its correct path. Repetitions count as separate token responses. An undecided test stops recovery without a score estimate; the error bounds below concern returning an incorrect estimate, not the probability of abstention.

\subsection{Known Response-Error Bound}
Assume each response, conditional on the complete preceding history, disagrees with $v$ with probability at most $0<\bar p_{\mathrm{flip}}<1/2$. Set
\begin{equation}
c=\frac{1-\bar p_{\mathrm{flip}}}{\bar p_{\mathrm{flip}}},\qquad
h=\left\lceil\frac{\log(M/\delta)}{\log c}\right\rceil.
\end{equation}
For repeated responses $Y_j\in\{0,1\}$, accumulate $S_n=\sum_{j=1}^n(2Y_j-1)$. Return 1 at $S_n\ge h$ and 0 at $S_n\le-h$. If neither boundary is reached by the response cap, abstain. This is a sequential likelihood-ratio rule~\cite{wald1945sequential} with a threshold chosen for a nonasymptotic, recovery-wide guarantee.

\begin{proposition}[Error control for adaptive recovery]\label{prop:noisy}
Under these assumptions, substituting the sequential rule for every logical comparison gives probability at most $\delta$ of returning an incorrect score pair or tie-breaking priority.
\end{proposition}
\begin{proof}
For $v=0$, the process $c^{S_n}$ is a nonnegative supermartingale starting at 1, since
\[
\mathbb E[c^{2Y_j-1}\mid\mathcal F_{j-1}]
\le\bar p_{\mathrm{flip}}c+(1-\bar p_{\mathrm{flip}})c^{-1}=1.
\]
The maximal inequality bounds the probability of reaching the wrong boundary by $c^{-h}\le\delta/M$. For $v=1$, apply the same argument to $c^{-S_n}$. Condition on the correct prefix of the adaptive controller and sum over the at most $M$ possible first erroneous comparisons. If none occurs, the deterministic recovery proof applies; stopping at the cap produces no incorrect estimate.
\end{proof}

For independent flips with actual probability $0<p_{\mathrm{flip}}<1/2$, the uncapped test has error probability and expected response count
\begin{equation}
\begin{aligned}
e_h(p_{\mathrm{flip}})&=\frac{1}{1+((1-p_{\mathrm{flip}})/p_{\mathrm{flip}})^h},\\
\mathbb E[N]&=\frac{h}{1-2p_{\mathrm{flip}}}(1-2e_h(p_{\mathrm{flip}})).
\end{aligned}
\end{equation}
The absorbing random walk on $\{-h,\ldots,h\}$ gives these expressions. At $M=50$, $\delta=0.01$, and $\bar p_{\mathrm{flip}}=p_{\mathrm{flip}}=0.1$, $h=4$ gives about five responses per comparison. Odd-$r$ independent majority voting has error at most $M\exp[-2r(1/2-\bar p_{\mathrm{flip}})^2]$; $r=27$ suffices, costing at most 1,350 responses.

\subsection{Unknown Response-Error Rate}
Suppose repetitions of a query are independent Bernoulli observations with fixed mean $p$. At sample sizes $n_j=2^j$, use
\begin{equation}
\begin{aligned}
\alpha_j&=\frac{\delta}{M j(j+1)},&
r_j&=\sqrt{\frac{\log(2/\alpha_j)}{2n_j}},\\
I_j&=[\widehat p_j-r_j,\widehat p_j+r_j]\cap[0,1].
\end{aligned}
\end{equation}
Hoeffding's inequality and $\sum_{j\ge1}\alpha_j=\delta/M$ give simultaneous coverage at all scheduled looks. Return 1 only when $\inf I_j>1/2$, or 0 only when $\sup I_j<1/2$. This implements the time-uniform coverage principle~\cite{howard2021timeuniform} through an elementary allocation of error probabilities. If the majority direction equals $v$ for each logical query, the same first-error argument bounds recovery error by $\delta$. Our implementation uses $j=1,\ldots,12$ and otherwise abstains. At a uniformly randomized token tie, $p=1/2$ provides no majority direction; increasing the cap cannot justify a fixed-priority conclusion.

\subsection{Random Token Tie Breaking and Changing Scores}
Assume fixed scores, independent uniform selection between tied maximizers, exact strict comparisons, and no third-token winner or tie. Repeat each query $r$ times, returning 1 only if all responses are 1, and returning 0 otherwise. This emulates token-0 priority, with exact strict comparisons and tie error probability $2^{-r}$. Thus $r=\lceil\log_2(M/\delta)\rceil$ bounds recovery error by $\delta$. For $M=50$ and $\delta=0.01$, $r=13$ gives at most 650 responses. The recovered priority is the chosen emulation convention, not an intrinsic priority of the randomized service.

These guarantees do not follow from a marginal disagreement rate alone. An error shared by all repetitions can make every response support the same wrong direction. Nor is an unspecified changing score vector equivalent to response flips: if observations use $q+\xi$, the representation $(q+c,\xi-c)$ gives identical observations for any fixed vector $c$. Identifying a particular unperturbed score therefore requires a restriction on the perturbation or a stable execution condition. The controlled batch experiment in Section~\ref{sec:noise} tests the latter condition.

\section{Logit Anchors and Output-Head Extraction}\label{app:head}
Following the factorization argument of Carlini et al.~\cite{carlini2024stealing}, let $W\in\R^{V\times d}$ be a bias-free linear output head and $X\in\R^{d\times n}$ its hidden states. Complete relative logits give $D=CWX$, with $C$ from~\eqref{eq:decomposition}. If $CW$ and $X$ have rank $d$, a rank factorization $D=\widehat B\widehat X$ satisfies $\widehat B=CWG$ and $\widehat X=G^{-1}X$ for an invertible $G$.

Write $\boldsymbol\mu=\mathbf1^\top WX/V\in\R^{1\times n}$ for the row of per-input mean logits. Each entry is determined by an original-logit anchor and its relative coordinates using~\eqref{eq:decomposition}. Then
\begin{equation}
\begin{aligned}
u&=\boldsymbol\mu\widehat X^\top(\widehat X\widehat X^\top)^{-1},\\
\widehat W&=\widehat B+\mathbf1u=WG.
\end{aligned}
\label{eq:head-completion}
\end{equation}
Since $\boldsymbol\mu=(\mathbf1^\top WG/V)\widehat X$ and $\widehat X$ has full row rank, $u=\mathbf1^\top WG/V$. The identity $C+\mathbf1\mathbf1^\top/V=I$ then gives $\widehat W=WG$. Anchors may vary by input: each determines its corresponding mean. If centering reduces rank, factor the completed matrix $WX=D+\mathbf1\boldsymbol\mu$ instead.

These identities assume exact linear outputs; floating-point residuals, relative-coordinate errors, and conditioning affect their numerical accuracy. For held-out inputs, centered coordinates determine hidden coordinates and $u$ predicts the mean without their anchors.

\noindent\textbf{Token-only numerical reconstruction.}
For each non-winning token, bias bisection retains losing and winning endpoints, stopping at width $2^{-21}$ or adjacent binary32 values. The negative midpoint estimates its score relative to the unbiased winner. Rounding contributes additional error. If a token cannot be reached within the bias limit, we mark the input incomplete. The artifact includes per-coordinate rounding bounds and independent checks.

\noindent\textbf{Numerical diagnostics and evaluation cost.}
Across input sets 1--3, the rank-64 spectral gaps are $3.09\times10^4$, $1.07\times10^5$, and $4.05\times10^4$; relative singular-value thresholds from $10^{-4}$ to $10^{-7}$ all retain rank 64. Effective condition numbers are 1,340, 441, and 1,048. Recomputing private $Wh$ in FP64 gives RMSE $1.40\times10^{-6}$--$1.45\times10^{-6}$ against the recorded FP32 logits, a numerical reference for the prediction errors. The held-out mean standard deviations are 0.831, 1.038, and 1.132; the constant-mean baseline uses only recovered extraction means. Held-out evaluation, including its independently scored anchors, adds 228,475/456,905/228,494 decisions and 7,564/15,858/7,500 HTTP requests. Those queries are excluded from the extraction costs in Table~\ref{tab:head-extraction}.

\section{Reproducibility Details}\label{app:repro}

\noindent\textbf{Paired bias-control experiment.}
The behavior-preservation study uses input indices $0,4,8,12,16,21,25,29,33,37,42,46,50,54,58,63$ from each serving configuration and PCG64 seed 2026101201. Random controls choose eight distinct tokens from the complete vocabulary or top 64, then draw independent biases uniformly from $[-A,A]$, for $A=0.5,2,10$. Top tokens are ordered by decreasing original score, then increasing ID. For the two leading tokens $A,T$, set $d=q_A-q_T$ and $U=\operatorname{nextafter}_{32}(m,+\infty)-m$, where $m=\max(|q_A|,|q_T|)$. The targeted controls set only $b_T=\mathrm{RN}_{32}(d+kU/8)$, $k=-16,\ldots,16$. All 12,384 payloads are legal under $H=100$ and distinct within each input. The two tie-breaking rules are fixed offline comparison conventions. Complete probabilities use CPU FP32 softmax, with $L_1$ differences accumulated in FP64.

Among threshold-adjacent controls, the original rounded top-two margin is zero in 834 cases, at most one ULP in 1,404 positive-margin cases, and between one and four ULPs in 930 cases. All 369 greedy changes occur in the zero-margin group. Random controls have 299 zero-margin cases and 8,917 cases above 16 ULPs; neither group changes.

\noindent\textbf{Serving API configuration.}
The client sends HTTP requests to each server's generation endpoint.
vLLM: \texttt{/v1/completions}, \texttt{logit\_bias}, \texttt{choices[0].token\_ids}; SGLang: \texttt{/generate}, \texttt{sampling\_params.logit\_bias}, \texttt{output\_ids}; llama.cpp: \texttt{/completion}, \texttt{logit\_bias}, \texttt{tokens}. Each triple gives the endpoint, bias field, and returned IDs. Decimal token IDs map to exact binary32 bias values within $[-100,100]$; probabilities are disabled. The artifact specifies payloads, launch flags, source hashes, and intervals.

\noindent\textbf{Artifact availability.}
The artifact is not publicly available at the time of this preprint. We plan to release these materials in a future revision.

\end{document}